\documentclass[12pt]{article}
\usepackage{hyperref}
\usepackage{amsfonts}
\usepackage{amsmath, amsthm, amssymb}
\theoremstyle{plain}
\newtheorem{theorem}{Theorem}[section]

\newtheorem{proposition}[theorem]{Proposition}

\newtheorem{definition}[theorem]{Definition}
\usepackage[round]{natbib}
\usepackage{bm}
\usepackage{dsfont}
\usepackage{bbm}
\usepackage{graphicx}
\usepackage{color}
\usepackage[usenames,dvipsnames]{xcolor}
\usepackage{soulutf8}
\usepackage[onehalfspacing]{setspace}
\usepackage{pdfpages}
\usepackage{rotating}
\usepackage{lscape}
\usepackage[round]{natbib}
\usepackage{placeins}
\usepackage{hhline}
\usepackage{multicol}
\usepackage{graphicx}
\usepackage[utf8]{inputenc}
\usepackage{tabularx}
\usepackage{floatrow}
\newcolumntype{L}{>{\raggedright\arraybackslash}X}
\usepackage{paralist}
\renewenvironment{itemize}[1]{\begin{compactitem}#1}{\end{compactitem}}

\newcommand{\bctm}[2]{\big( \varthetavec, \allowbreak {#1},\allowbreak {#2}, \allowbreak\pi_\vartheta(\varthetavec)\big)}
\newcommand{\MC}[1]{\textbf{Manu: \textcolor{blue}{#1}}}

\def \dsR {\text{$\mathds{R}$}}

\def \avec {\text{\boldmath$a$}}    
\def \bvec {\text{\boldmath$b$}}    
\def \cvec {\text{\boldmath$c$}}    \def \mC {\text{\boldmath$C$}}
    \def \mD {\text{\boldmath$D$}}

    \def \mI {\text{\boldmath$I$}}
    
    \def \mK {\text{\boldmath$K$}}

\def \xvec {\text{\boldmath$x$}}    \def \mX {\text{\boldmath$X$}}
\def \yvec {\text{\boldmath$y$}}

\def \betavec         {\text{\boldmath$\beta$}}
\def \gammavec        {\text{\boldmath$\gamma$}}

\def \varthetavec     {\text{\boldmath$\vartheta$}}

\def \tauvec          {\text{\boldmath$\tau$}}

\def \omegavec        {\text{\boldmath$\omega$}}

\def \betatildevec         {\text{\boldmath$\tilde \beta$}}

\def \mSigma   {\mathbf{\Sigma}}

\usepackage[font={stretch=1.0,footnotesize}]{caption}
\newcommand{\blind}{1}
\allowdisplaybreaks
\begin{document}

\def\spacingset#1{\renewcommand{\baselinestretch}%
{#1}\small\normalsize} \spacingset{1}

\if1\blind
{
  \title{\bf Bayesian Conditional Transformation Models}
  \author{Manuel Carlan\thanks{
    The work of Manuel Carlan was supported by DFG via the research training group 1644.} ,\hspace{.2cm} Thomas Kneib\thanks{
    Thomas Kneib received financial support from the DFG within the research project KN 922/9-1.}\\
    Chair of Statistics, University of G\"{o}ttingen\\
    and\\
    Nadja Klein\thanks{\protect\linespread{1}\protect\selectfont
    Nadja Klein gratefully acknowledges funding from the DFG through the Emmy Noether grant\newline\hspace*{1.7em}KL 3037/1-1. The authors would like to thank the Editor, Associate Editor and two referees for\newline\hspace*{1.7em}many valuable comments that lead
to a significant improvement of our original submission.\newline\hspace*{1.7em}nadja.klein@hu-berlin.de}\\
    Chair of Statistics and Data Science, Humboldt-Universit\"{a}t zu Berlin
    }
      \date{}
  \maketitle
  \thispagestyle{empty}
} \fi

\if0\blind
{
  \bigskip
  \bigskip
  \bigskip
  \begin{center}
    {\LARGE\bf Bayesian Conditional Transformation Models}
\end{center}
  \medskip
} \fi
%\bigskip
%\pagestyle{empty}
%\renewcommand{\baselinestretch}{1.0}
\begin{abstract}
\noindent Recent developments in statistical regression methodology shift away from pure mean regression towards distributional regression models. One important strand thereof is that of conditional transformation models (CTMs). CTMs infer the entire conditional distribution directly by applying a transformation function to the response conditionally on a set of covariates towards a simple log-concave reference distribution. Thereby, CTMs allow not only variance, kurtosis or skewness but the complete conditional distribution to depend on the explanatory variables. We propose a Bayesian notion of conditional transformation models (BCTMs) focusing on exactly observed continuous responses, but also incorporating extensions to randomly censored and discrete responses. Rather than relying on Bernstein polynomials that have been considered in likelihood-based CTMs, we implement a spline-based parametrization for monotonic effects that are supplemented with smoothness priors. Furthermore, we are able to benefit from the Bayesian paradigm via easily obtainable credible intervals and other quantities without relying on large sample approximations. A simulation study demonstrates the competitiveness of our approach against its likelihood-based counterpart but also Bayesian additive models of location, scale and shape and Bayesian quantile regression. Two applications illustrate the versatility of BCTMs in problems involving real world data, again including the comparison with various types of competitors.

\vspace{0.5cm}

\end{abstract}

\clearpage
\pagenumbering{arabic} \setcounter{page}{1}

\spacingset{1.9}
\setlength{\abovedisplayskip}{0.2cm}
\setlength{\belowdisplayskip}{0.2cm}
\section{Introduction}\label{sec:intro}
Regression is omnipresent in many statistical applications and an ongoing field in recent research on statistical methods. While, in principle, interest always lies in describing the conditional distribution $\mathbb{P}_{Y|\mX=\xvec}$ of a response $Y$ given a set of explanatory variables $\mX$ with observed realisations $\xvec$, most traditional approaches target the conditional expectation $\mathbb{E}(Y|\mX = \xvec)$ as the only characteristic of interest \citep[e.g.~generalized linear or additive models;][]{nelder1972generalized,hastie1990generalized}.
One way to abolish this often unwarranted simplification are generalized additive models for location scale and shape \citep[GAMLSS;][]{rigby2005}  allowing for flexible relationships between the covariates and all parameters of the response distribution via flexible additive predictors. In this framework, the researcher can select from a diverse set of parametric distributions for discrete, continuous, mixed and multivariate response distributions \citep{klein2015bayesian,KleKneKlaLan2015}. However, deciding on a parametric response distribution can quickly become a burden as it imposes strong assumptions on the model if not done with great care. One approach that does not entail a fixed parametric form is quantile regression introduced by \cite{KoeBas1978}. Since a distribution is satisfyingly captured by a dense grid of quantiles, each of these quantiles is modelled linearly or additively through covariates \citep{horowitz2005nonparametric}. Bayesian versions were developed e.g.~by \cite{yu2001bayesian,waldmann2013bayesian}.

In contrast to all approaches mentioned so far, transformation models aim to infer the conditional distribution function directly. In an attempt to draw a bigger picture, we recapitulate a brief history of transformation models, while slightly focusing on Bayesian implementations. For a tour de force that concentrates on the frequentist perspective, see e.g.~\cite{hothorn2014,most2015conditional,hothorn2018}.
Every transformation model entails a monotonically increasing transformation function $h$ that acts on the response and is designed to reframe an unknown distribution $\mathbb{P}(Y \leq y)$ in terms of the transformation, i.e.~$\mathbb{P}(h(Y) \leq h(y))$. The advent of parametric transformation models goes back to the Box-Cox model \citep{box1964analysis} which ignited an area of active research that is still lit to this day.  One approach that avoids strong assumptions on the parametric form of the transformation function was introduced by \cite{cheng1995analysis}. It inspired plenty of models that share the estimation of a linear transformation function $h(y|\xvec) = h_Y(y)- \xvec^\top\betavec$ where the baseline transformation $h_Y(y)$ is estimated semiparametrically in conjunction with a linear, covariate-dependent shift $\xvec^\top\betavec$. Prominent representatives are the proportional odds or the proportional hazards model. One of the first Bayesian transformation models was proposed by \cite{pericchi1981bayesian}. \cite{mallick2003bayesian} model  the transformation function $h$ semiparametrically using (Bayesian) Bernstein polynomials and P\'{o}lya trees for the estimation of accelerated failure time (AFT) models among others.  \cite{song2012semiparametric} use  Bayesian P-splines \citep{lang2004bayesian} for transformation models with additive shift effects and a Gaussian reference distribution. Besides continuous responses, \cite{james2021bayesian} allow for discrete ordered and mixed discrete/continuous outcomes in conjunction with linear covariate effects.

Although very powerful in a lot of applications, transformation models of this type are considerably hindered by the additivity assumption on the scale of the transformation function $h(y|\xvec) = h_Y(y) + h_{\xvec}(\xvec)$ where the explanatory variable can only contribute a shift of the baseline transformation $h_Y$ and can therefore influence the conditional location parameter only.  One modern example that includes linear interactions of covariates and gained a lot of attention is distribution regression \citep{chernozhukov2013inference}. Here, in the context of counterfactuals, the conditional transformation function $h(y|\xvec)=h_Y(y) - \xvec^\top \betavec(y)$ is supplemented with varying-coefficient type interactions $\xvec^\top \betavec(y)$ where the varying coefficients $\betavec(y)$ are estimated on basis of $\mathbb{P}(Y \leq y | \mX=\xvec)=\mathbb{E}(\mathbbm{1}(Y\leq y) |\mX=\xvec)$. This connection allows to account for heteroskedasticity or other patterns that vary with the covariates.
An even more flexible variant comes with conditional transformation models (CTMs) as introduced by \cite{hothorn2014} which share the same goal and aim to obtain an estimator for the whole conditional  distribution function.

In this article, we propose the class of Bayesian conditional transformation models (BCTMs). BCTMs can be understood as a Bayesian interpretation of the likelihood-based CTMs of \cite{hothorn2018} via the most likely transformation (MLT) model. All three models have in common that they target the direct estimation of the distribution function of a response $Y$ conditional on a set of covariates $\mX =\xvec$ by means of estimating the conditional transformation function $h(y|\xvec)$. Yet, Bayesian inference based on Markov chain Monte Carlo (MCMC) simulations additionally allows us to obtain exact inferences on all quantities of interest without relying on large sample approximations or bootstrap procedures. This can be of particular value in scenarios with smaller samples where the parameters themselves are of secondary interest compared to complex transformations thereof. As the MLT model, the BCTM can be applied to discrete and continuous responses in the presence of random censoring, but additionally includes smoothness penalties for high-dimensional effects induced by the prior supporting stable function estimates. Bayesian principles in connection with the modularity of  implementation make it straightforward to expand BTCMs towards more complex prior structures enabling effect selection or different shrinkage properties for example. The idea of using monotonic P-splines for parametrizing transformation functions has been explored before \citep[see e.g.,][]{song2012semiparametric,tang2018semiparametric}, but our approach is innovative in a Bayesian setting with higher-dimensional interactions involving $y$, where curvatures are often complex and where  control over the penalization mechanism can  contribute to a better understanding of the model.
To summarize,  we
\begin{compactitem}
\item introduce BCTMs as a new model class,
\item apply a B-spline basis in conjunction with reparametrized basis coefficients to impose monotonicity  on the conditional transformation function in the $y$ direction (opposed to the MLT model which relies on simple Bernstein polynomials),
\item supplement the unreparameterized vector of basis coefficients with a partially improper multivariate Gaussian prior that enforces smoothness towards a straight line both for monotonic and for unrestricted nonlinear effects,
\item develop Bayesian posterior estimation based on Hamiltonian Monte Carlo \citep[HMC;][]{neal2011mcmc,betancourt2017conceptual} using the highly-efficient No-U-Turn Sampler \citep[NUTS,][]{hoffman2014no} for the vector of basis coefficients,
\item implement Bayesian model selection,
\item evaluate the distribution recovery ability and  validity of credible intervals for BCTMs in different simulations and compare them to its main competitors, and
\item demonstrate different aspects and the practical relevance of BCTMs in applications on cholesterol levels from the Framingham heart study and on leukemia survival times.
\end{compactitem}
The rest of the paper is structured as follows: Sec.~\ref{sec:CTM} introduces BCTMs as a model class consisting of several building blocks including prior assumptions and theoretical properties. Sec.~\ref{sec:mcmc} describes posterior estimation including Bayesian model selection. Sections~\ref{sec:apps} and \ref{sec:sims}  contain simulations and applications, respectively. Sec.~\ref{sec:summary} provides a brief review of our findings and proposes several directions for future research. The Supplement contains proofs of theoretical results as well as additional results for the simulations, an additional application on lung cancer survial times from the Veteran's Administration Lung Cancer Trial and further details.

\section{Bayesian Conditional Transformation Models}\label{sec:CTM}

In a CTM, the cumulative distribution function (CDF) of a response $Y\in\mathcal{S}\subset\dsR$ conditional on a set of covariates $\mX$ is specified via
\begin{align}
F_{Y|\mX={\xvec}}(y)=P(Y \leq y | \mX = \xvec) &= P(h(Y|\xvec) \leq h(y|\xvec))=  F_Z(h(y|\xvec)),\label{eq:ctm}
\end{align}
where the covariate-dependent function $h(y|\xvec):\mathcal{S}\to\dsR$ is assumed to be monotonically increasing in $y$ to transform the response such that it follows a pre-specified reference distribution with continuous distribution function $F_Z : \mathbb{R} \mapsto [0,1]$. The reference distribution is independent of $\xvec$ and does not contain any unknown parameters to be estimated. %In this way, an estimate $\hat{F}_{Y|\mX}$ of a possibly complex conditional cumulative distribution function is obtained by estimating $h(y|\xvec)$ in conjunction with choosing $F_Z$.
In this way, a CTM is  characterized by the choice of the reference distribution $F_Z$ and a suitable parameterisation of the transformation function $h(y|\xvec)$ such that an estimate of the latter yields an estimate of a possibly complex conditional cumulative distribution function (cCDF) $\hat{F}_{Y|\mX}$. We will discuss both ingredients in more detail below. Naturally, distinctive characteristics of the transformation function such as monotonicity and smoothness are mirrored in $F_{Y|\mX = \xvec}$, which is why $h(y|\xvec)$ has to be modelled with great care.

Following \cite{hothorn2014}, we assume an additive decomposition on the scale of the transformation function into $J$ partial transformation functions, i.e.
\begin{equation}
\begin{aligned}
h(y|\xvec) &= \sum_{j=1}^J h_j(y|\xvec), &j=1,\ldots, J,
\end{aligned}\label{eq:add_decomp}
\end{equation}
where $h_j(y|\xvec)$, in the broadest sense, can be understood as response-covariate interactions that are monotone only in direction of $y$. To ensure identifiability, the partial transformation functions involving nonlinear terms are centered around zero, resulting in the additive decomposition $h(y|\xvec) = \beta_0 + \sum_{j=1}^J h_j(y|\xvec)$ with overall intercept $\beta_0$ of the conditional transformation function, which we will notationally suppress for most of what follows.

In light of \eqref{eq:add_decomp}, it is important to stress that additivity of the transformation function is assumed on the transformed scale, i.e.\ there is no explicit differentiation between signal and noise as in Gaussian regression models with separable error term. Hence, CTMs come with the benefit of a straightforward entry point to modelling all moments of the response distribution implicitly as functions of $\xvec$. In the realm of CTMs, the flexibility of $h_j(y|\xvec)$ constitutes the scope of the impact a covariate is admitted to have on the whole cCDF.

We assume that each of the $J$ partial transformation functions $h_j(\cdot|\xvec)$, can be approximated by a linear combination of basis functions such that
$h_j(y|\xvec)  = \cvec_j(y,\xvec)^\top \gammavec_j$,
where $\gammavec_j$ is a vector of basis coefficients. Later we assume monotonicity of each partial transformation function in $y$, i.e.\ $h_j^\prime(y|\xvec) = \frac{\partial h_j(y|\xvec)}{\partial y} \ge 0$, which is sufficient but not necessary for an overall monotonic transformation function $h(y|\xvec)$. The complete transformation function and its derivative with respect to $y$ are now given by
\begin{align}\label{eq:full_h}
\begin{split}
h(y|\xvec) = \cvec(y,\xvec)^\top \gammavec, \quad
h^\prime(y|\xvec) = \cvec^\prime(y,\xvec)^\top \gammavec,
\end{split}
\end{align}
with bases $\cvec(y,\xvec) = (\cvec_1(y, \xvec)^\top, \ldots, \cvec_J(y,\xvec)^\top)^\top$ and $\cvec^\prime(y,\xvec) = (\cvec^\prime_1(y, \xvec)^\top, \ldots,\allowbreak \cvec^\prime_J(y,\xvec)^\top)^\top$ for the transformation function and its derivative, respectively, and the stacked vector of all basis coefficients
$\gammavec = (\gammavec_1^\top, \ldots, \gammavec_J^\top)^\top.$

In the following subsection, we introduce a generic and flexible joint basis for $c_j(y,\xvec)$ that does not entail strict assumptions about the relationship between the moments of the response distribution and the respective covariates. Prior distributions, specific bases for the covariate effects, the choice of the references distribution $F_Z$ and a formal definition of our BCTM are covered in Sec.~\ref{sec:priors} to \ref{sec:formaldef}.

\subsection{Generic conditional transformation functions}\label{sec:generic_ct}
Let $\avec_j(y)$ and $\bvec_j(\xvec)$ denote vectors containing basis function evaluations $B_{j1d_1}(y), d_1 = 1, \ldots D_1$ and $B_{j2d_2}(\xvec), d_2 = 1,\ldots, D_2$ for the response and the covariates, respectively, such that
$\avec_j(y)^\top = (B_{j11}(y),\ldots,B_{j1D_1}(y))$, and $\bvec_j(\xvec)^\top = (B_{j21}(y),\ldots,B_{j2D_2}(y))$.
Denoting by $\otimes$ the usual Kronecker product, we then obtain the most general form of partial transformation function in a BCTM as $\cvec_j(y,\xvec)^\top =  (\avec_j(y)^\top \otimes \bvec_j(\xvec)^\top)^\top$, leading to
\begin{align}
\begin{split}
h_j(y |\xvec) &= \cvec_j(y, \xvec)^\top \gammavec_j = (\avec_j(y)^\top \otimes \bvec_j(\xvec)^\top)^\top \gammavec_j = \sum_{d_1=1}^{D_1} \sum_{d_2=1}^{D_2} \gamma_{jd_1d_2} B_{j1d_1}(y)B_{j2d_2}(\xvec)\\
h^\prime_j(y |\xvec) &= \cvec^\prime_j(y, \xvec)^\top \gammavec_j = (\avec_j^\prime(y)^\top \otimes \bvec_j(\xvec)^\top)^\top \gammavec_j = \sum_{d_1=1}^{D_1} \sum_{d_2=1}^{D_2} \gamma_{jd_1d_2} B^\prime_{j1d_1}(y)B_{j2d_2}(\xvec).
\end{split}\label{full_bctm}
\end{align}
Essentially, the Kronecker product establishes a parametric interaction by forming pairwise products of the basis functions $B_{j1d_1}(y)$ and $B_{j2d_2}(\xvec)$. The derivative with respect to $y$ is therefore also a tensor product involving the differentiated basis functions $B^\prime_{j1d_1}(y)=\partial B_{j1d_1}(y)/\partial y $. Specific restrictions on the two components of the tensor product lead to interesting special cases of the partial transformation function:
\begin{itemize}
    \item Setting $\avec_j(y)\equiv1$ leads to simple shift effects that only depend on the covariates.
    \item Setting $\bvec_j(\xvec)\equiv1$ yields an effect of only $y$ that induces changes of the distributional shape (up to other effects).
    \item Linear effects $\avec_j(y)=(1,y)$ induce varying coefficient type effects where covariate effects linearly interact with the responses and  the response takes the role of the interaction variable while the covariates are the effect modifiers.
\end{itemize}

Restricting our generic model to
\begin{align}\label{eq:spec_trans}
h(y|\xvec) = h_Y(y) + h_{\xvec}(\xvec),
\end{align}
leads to a location-shift transformation model that comprises various earlier transformation models as special cases. In this case, only the location of the transformed response depends on the covariates via $h_{\xvec}(\xvec)$ and higher moments are captured unconditionally by the monotonic transformation $h_Y(y)$. This model type is parametrized by restricting the joint basis to $\cvec(y,\xvec)^\top = ((\avec(y)^\top\otimes 1)^\top, (1 \otimes \bvec(\xvec)^\top)^\top)=(\avec(y)^\top, \bvec(\xvec)^\top)^\top$ resulting in the shift transformation model  $\cvec(y,\xvec)^\top \gammavec = \avec(y)^\top \gammavec_1 + \bvec(\xvec)^\top \gammavec_2$.

We are relying on B-splines for the response dimension while various alternatives are available for the covariate dimension (see Sec.~\ref{sec:special_cases} for details). The choice of B-splines for representing $\avec_j(y)$ is mainly determined by the availability of suitable reparameterisations of the corresponding basis coefficients that ensure monotonicity along $y$ in the tensor product for the partial response transformations and well-studied smoothness properties. More precisely, we follow \citet{pyawood} and reparameterize the $D=D_1D_2$ dimensional basis vector
$\gammavec_j=(\gamma_{j11},\ldots,\gamma_{j1D_2},\gamma_{j21},\ldots,\gamma_{jD_1D_2})^\top$ in two steps. First, we set $\gammavec_j= \mSigma_j \betatildevec_j$, where  $\mSigma_j = \mSigma_{D_1} \otimes \mI_{D_2}$, $\mI_{D_2}$ is an identity matrix of size $D_2$, $\mSigma_{D_1}$ is a lower triangular matrix of size $D_1$ with $\Sigma_{D_1,kl}=0$ if $k < l$ and $\Sigma_{D_1,kl} = 1$ if $k \geq l$, and the vector $\betatildevec_j$ is
\begin{align}\label{eq:bt2}
\betatildevec_j = (\beta_{j11}, \ldots, \beta_{j1D_2}, \exp(\beta_{j21}), \ldots \exp(\beta_{j2D_2}), \ldots ,\exp(\beta_{jD_1 1}),\ldots, \exp(\beta_{jD_1 D_2}))^\top.
\end{align}
Starting with a vector $\betavec_j=(\beta_{j11}, \ldots, \beta_{j1D_2}, \beta_{j21}, \ldots \beta_{j2D_2}, \ldots ,\beta_{D_1 1},\ldots, \beta_{jD_1 D_2})^\top\in\mathbb{R}^D$ of unconstrained parameters, these choices ensure that the vector of basis coefficients $\gammavec_j$ is strictly increasing along the response dimension $y$ which, in turn, implies a tensor product effect that is monotonically increasing along $y$. The complete model vectors of basis coefficients are then given by $\betavec = (\betavec_1^\top, \ldots, \betavec_J^\top)^\top$ and $\betatildevec =(\betatildevec_1^\top, \ldots, \betatildevec_J^\top)^\top$, while the overall model matrix $\mSigma$ is block diagonal with $\mSigma_1, \ldots, \mSigma_J$ as diagonal elements. We formalize the monotonicity of $h(y|\xvec)$ along the $y$ dimension in the following theorem. 
\begin{theorem}[Monotonically increasing transformation function along $y$]\label{theo1}
Let $h(\cdot|\xvec):\mathcal{S}\to\dsR$ be the transformation function \eqref{eq:add_decomp} with basis representation \eqref{eq:full_h}  and partial transformation functions $h_j$ as in \eqref{full_bctm}. Let furthermore $\gammavec_j=\mSigma_j\tilde\betavec_j$ with $\tilde\betavec_j$ as in \eqref{eq:bt2} and  $\mSigma_j = \mSigma_{D_1} \otimes \mI_{D_2}$ as defined above. Then,  $h(\cdot|\xvec)$ is monotonically increasing, that is for all $y_1,y_2\in\mathcal{S}$ with $y_1<y_2$ we have $h(y_1|\xvec)\leq h(y_2|\xvec)$.
\end{theorem}
\noindent A proof of Theorem \ref{theo1} can be found in the Supp.~Part A.

In contrast to the MLT model introduced by \cite{hothorn2018} that uses Bernstein polynomials as a basis for nonlinear effects, the BCTM is supplemented with a smoothness-inducing penalty through its prior and is therefore in principle less restrained regarding the number of model terms $J$ and functional complexity in direction of the covariates.

Of course, other basis function representations than B-splines are immediately conceivable for $\avec_j(y)$. The main requirements for a suitable specification include the ability to incorporate monotonicity constraints, the analytical availability of the basis functions and their derivatives, and the numerically stable evaluation of these. While B-splines fulfill these requirements, investigating other choices and their properties is a promising avenue for future research.

\subsection{Prior specifications}\label{sec:priors}

Overfitting of unregularized splines can be avoided in our Bayesian framework by enforcing smoothness and regularization through shrinkage priors. For the special case of B-splines, Bayesian P-splines assign multivariate Gaussian priors to the regression coefficient vectors. We follow \citet{kneib2019modular} and adopt this principle to tensor product terms such that the prior for the coefficient vector $\betavec_j$ associated with one of the partial transformation functions $h_j$ in \eqref{full_bctm} is multivariate Gaussian with expectation zero and precision matrix
\begin{align}\label{eq:priorprec}
\mK_j\equiv\mK_{j}(\tau^2_j,\omega_j) =\frac{1}{\tau_{j}^2}\Bigl\lbrack \omega_j(\mK_{1j}\otimes\mI_{D_2}) + (1-\omega_j)  (\mI_{D_1}\otimes\mK_{2j})\Bigr\rbrack,
\end{align}
where $\mK_{j1}$ and $\mK_{j2}$ are potentially rank deficient prior precision matrices of dimensions ($D_1\times D_1$) and ($D_2\times D_2$), respectively, controlling the type of smoothness required along the response and the covariate dimension, respectively. For the response dimension, we set $\mK_{1j} = \mD_{1j}^\top \mD_{1j}$ where $\mD_{1j}$ is a $(D_1-2) \times D_1$ partial first difference matrix consisting only of zeros except that $\mD_{j}[d_1,d_1+1] = - \mD_{j}[d_1,d_1+2]=1$ for $d_1=1,\ldots, D_1-2$ \citep{pyawood,pya2010additive}.  The prior precision matrix $\mK_{2j}$ shrinks in the direction of the respective covariate and the specific choice depends on the considered covariate effect of interect (see Sec.~\ref{sec:special_cases} for some examples). For monotonic nonlinear effects, the resulting penalty is quadratic in the (non-exponentiated) parameters $\betavec_j$. This corresponds to log differences in $\gamma_{jd_1d_2}$ for $d_1>2$, such that a first order random walk prior penalizes the squared differences between adjacent $\beta_{jd_1 d_2}$, resulting in shrinkage towards a straight line, similar to second order random walk penalties for univariate P-splines \citep{pyawood}.  The complete prior precision matrix $\mK\equiv\mK(\tauvec^2,\omegavec)$ is given as the block diagonal matrix with matrices $\mK_{j}$ as diagonal elements. We formalize the prior for $\gammavec_j$ in the following proposition.
\begin{proposition}[Prior for $\gammavec_{j}$]
Let $p_\beta(\betavec_j|\tau_j^2)\propto\exp\left(-\tfrac{1}{2\tau_j^2}\betavec_j^\top\mK_j^{-}\betavec_j\right)$ be the partially improper multivariate Gaussian prior with generalized inverse $\mK_j^{-}$ of the prior precision matrix in \eqref{eq:priorprec}. Assume furthermore for notational simplicity that $D_2=1$ such that $\gamma_{jd1d2}\equiv\gamma_{jd1}$. Then, the prior for $\gammavec_j$ is given by
\[
p_\gammavec(\gammavec_j|\tau^2)\propto p_\betavec\left(\gamma_{j1},\log(\gamma_{j2}-\gamma_{j1}),\ldots,\log(\gamma_{jD}-\sum_{i=1}^{D-1}(D-i)\gamma_{ji})\right)\prod_{k=1}^D\frac{1}{\gamma_{jk}-\sum_{i=1}^{k-1}(k-i)\gamma_{ji}}
\]
\end{proposition}
\begin{proof}
The proof follows directly by applying the multivariate change of variable theorem twice to the transformation $g_2\circ g_1:\dsR^D\to\dsR^D$ with $\tilde\beta_1=g_1(\beta_{j1})=\tilde\beta_{j1}$, $\tilde\beta_{jk}=g_1(\beta_{jk})=\exp(\beta_{jk})$, $k=2,\ldots,D$ and $\gamma_{jk}=g_2(\tilde\beta_{jk})=\sum_{i=1}^k\tilde\beta_{ji}$, $k=1,\ldots,D$.
\end{proof}
The amount of smoothness induced by the precision matrix (\ref{eq:priorprec}) is controlled by the overall smoothing variance $\tau_{j}^2>0$ and the weight parameter $\omega_j\in[0,1]$. Following \citet{kneib2019modular}, we assume a discrete prior for the latter which has the advantage that generalized determinants of $\mK_j$ can be pre-computed which considerably facilitates the numerically efficient implementation while still enabling anisotropic amounts of smoothness along the response and the covariate dimension. A uniform prior on a moderate number of equi-spaced values is used as a default for $\omega_j$. For the smoothing variance $\tau_j^2$, we consider two alternatives: Standard inverse gamma (IG) priors $\tau_{j}^2 \sim \mathrm{IG}(a_{j},b_{j})$ where the hyperparameters are chosen among popular combinations such as $a_{j}=1$, $b_{j}=0.001$ to mimic a weakly informative setting, and scale-dependent (SD) hyperpriors as suggested in \citet{KleKne2016}. The latter results in a Weibull prior for $\tau_j^2$ with shape parameter $0.5$ and scale parameter $\theta$ determined from a scaling criterion on expected effect sizes. We transfer this concept to partially monotonic tensor product effects where, to achieve numerical stability, it is important to control the variation of those parameters exponentiated in (\ref{eq:bt2}).  More precisely, we consider the scaling criterion
$
 \mathbb{P}\left(
 \max_{d_1=2,\ldots,D_1, d_2=1,\ldots,D_2}
 |\beta_{jd_1,d_2}|
 \le c\right) = 1-\alpha
$
with user-specified values for $c$ and $\alpha$. To determine the marginal prior distribution of $\tau_j^2$ required to evaluate the scaling criterion, we follow a simulation-based approach to marginalize out any additional hyperparameters. From the support of the exponential function, $c=3$ and $\alpha=0.01$ are useful standards also used later in our empirical studies.

In a last step, we collect all model parameters in the vector
\begin{align*}
\varthetavec= (\beta_0,\betavec_1, \ldots, \betavec_J, \tau_{1}^2, \ldots, \tau_{J}^2, \omega_1, \ldots, \omega_J)^\top = (\betavec^\top, (\tauvec^2)^\top, \omegavec^\top)^\top
\end{align*}
with joint prior $\pi_{\varthetavec}(\varthetavec)$ which can be factorized into products of the individual priors. The coefficient $\beta_0$ denotes the intercept of the model, while $\betavec$, $\tauvec^2$, $\omegavec$ are used to denote all basis coefficients, smoothing variances and anisotropy weights, respectively.

\subsection{Bases for the covariate effects}\label{sec:special_cases}

We highlight special cases of bases  $\bvec(\xvec)$ relevant for our applications:
\begin{itemize}
\item \textit{Linear effects.} The basis for linear effects of covariates $x_1, \ldots, x_p$ collected in $\xvec$ is $\bvec_j(\xvec)^\top=(x_1, \ldots, x_p)^\top$ and we use a non-informative prior with $\mK_{2j} = \bm{0}$. This also applies to the overall intercept $\beta_0$ when centering the partial transformation functions.
\item \textit{Random effects.} Random effects (or frailties) are based on a grouping indicator $g \in \{1, \ldots, G\}$. The resulting $G$-dimensional basis vector $\bvec_{j}(g)$ has entry one if $y$ belongs to group $g$ and zero otherwise and we set $\mK_{2j} = \mI_G$ for i.i.d.\ random effects.
\item \textit{Discrete spatial effects.} Similar to random effects, a spatial effect of a discrete spatial variable $s \in \{1, \ldots, S\}$ is constructed as an indicator  with entries in the $S$-dimensional basis vector
$\bvec_{j}(s)$ set to one if $y$ belongs to region $s$ and zero otherwise. We induce  spatial smoothing in form of a Gaussian Markov random field  \citep[GMRF][]{rue2005gaussian} prior. The  precision matrix $\mK_{2j}$ reflects the spatial orientation of the data, i.e.~we define two regions as neighbours if they share a common border.
\end{itemize}
%Random and spatial effects pose additions to the likelihood-based MLT model \citep[\textbf{mlt},][]{hothorn2017package} not implemented therein.
For an overview on possible bases for discrete (count) responses we refer the reader to \cite{bdctm}.

\subsection{Choice of the reference distribution}\label{sec:reference}
As already stated, it is the task of the conditional transformation function $h(y|\xvec)$ to transform the response values conditionally on the explanatory variables $\xvec$ such that they follow the reference distribution $F_Z$. In that light, $F_Z$ plays a similar role as the inverse known link function in prominent model classes such as generalized linear models, but is less restrictive in the sense that the resulting conditional distribution does not have to be of known type. From a modelling perspective, it guarantees that the resulting estimated conditional density function integrates to one without requiring complex constraints. Note that no unknown parameters are included in $F_Z$ and no restrictions besides continuity and log-concavity of the reference density $f_Z(y) = \frac{\partial F_Z(y)}{\partial y}$ are required. In theory, any cCDF can be represented as a BCTM when the transformation function is chosen flexible enough. However, in practice the actual ability to represent various types of cCDFs is limited by the choices made for parameterizing the transformation function. For example, when the reference distribution has light tails, one requires considerable flexibility in the transformation function to enable the representation of heavy-tailed distributions. Similarly, restricting the shape of the influence that the covariates can have on the transformation function also imposes restrictions on the cCDFs that can be generated via a BCTM.

Other relevant aspects for the choice of the reference distribution entail (i) interpretation, (ii) convenience, and (iii) theoretical properties. For the sake of interpretation, it is advised to consider further characteristics such as skewness or positivity of $y$ when choosing $F_Z$. Prominent options also used in the applications in Sec.~\ref{sec:apps} are the standard normal CDF, $F_Z(z)=\Phi(z)$, the standard logistic CDF $F_Z(z)=F_{\text{SL}}(z)=(1+\exp(-z))^{-1}$ (leading to (non-)proportional odds models) and the minimum extreme value distribution, $F_Z(z)=F_{\text{MEV}}(z)=1-\exp(-\exp(z))$ (leading to (non-)proportional hazards models). Note that  simple transformation models of type \eqref{eq:spec_trans} are interpretative in the sense that the term $h_{\xvec}(\xvec)$ constitutes the log odds ratio if $F_Z(z)=F_\text{SL}(z)$ and the log hazards ratio if $F_Z(z)=F_\text{MEV}(z)$, a previous result we use in Sec.~\ref{sec:lung}.

The convenience argument favours distributions that are well studied on the one hand and numerically easy to handle on the other hand. Finally, certain properties of the resulting estimates also depend on the choice of the reference distribution. For example, restricting the reference distribution to have log-concave densities ensures that (under typical additional mild regularity conditions) that the MLE is unique and consistent which, in turn, often implies unimodal posteriors that are easier to explore with MCMC schemes. For an overview of the numerous possibilities of reference distributions that come with CTMs, see \cite{hothorn2018}.

\subsection{Transformation densities}\label{sec:trafo_ll}

In this section, we introduce the conditional transformation densities $f_Y(y|\betavec)$ given the  vector of basis coefficients $\betavec$ (before the reparameterization). To emphasize that  $\gammavec$ is a partially nonlinear reparameterization of $\betavec$, we write $\gammavec(\betavec)$.

\paragraph{Continuous responses} The density and log-density can easily be derived from equation \eqref{eq:ctm} together with the parametrization of $h$ and $h^\prime$ in equation \eqref{full_bctm} such that
\begin{align}\label{cont_ll}
f_Y(y|\betavec) &= \frac{\partial F_{Y|\mX={\xvec}}(y)}{\partial y} = f_Z\left(h(y|\xvec)\right)h^\prime(y|\xvec)  = f_Z(\cvec(y, \xvec)^\top \gammavec(\betavec))\cvec^{\prime}(y, \xvec)^\top \gammavec(\betavec),
\end{align}
%\begin{align*}
%\mathcal{L}(h; y,\xvec) &= f(h(y|))h^\prime(h|\xvec) \\
%\mathrm{log}(\mathcal{L}(h; y,\xvec)) &= \mathrm{log}(f(h(y|\xvec))) + \mathrm{log}(h^\prime(y|\xvec)),
%\end{align*}
where $f_Z$ denotes the density of the chosen reference distribution. In theory, for any absolute continuous response distribution $F_Y$ and reference distribution $F_Z$ with log-concave density $f_Z$, there exists a unique, monotonically increasing transformation function $h$, such that $F_{Y|\mX=\xvec}(\cdot)=F_Z(h(\cdot|\xvec))$ \citep[see Cor.~1 of][]{hothorn2018}.

It is important to note that both for univariate and bivariate effects involving $y$, the part of the effect that belongs to the null space of $\mK_j$ consists of all location shifts and linear effects in $y$. In the context of BCTMs with the popular choice $F_Z=\Phi$, this means that the penalty shrinks towards the Gaussian location-scale family. In other words, the null space of the rank-deficient precision matrix consists of all Gaussian conditional distribution functions. This observation can be put to use when considered from the perspective of SD priors for the variances as described in \cite{KleKne2016}.

\paragraph{Discrete ordinal responses} In case of discrete ordinal responses with a finite sample space where $Y \in \{y_1, \ldots, y_K\}$, the corresponding conditional density function is given by

\begin{align}\label{disc_ll}
\begin{split}
f_Y(y_k |\betavec)=
\begin{cases}
    F_Z(\cvec(y_1,\xvec)^\top \gammavec(\betavec)) & k = 1\\
     F_Z(\cvec(y_k,\xvec)^\top \gammavec(\betavec)) - F_Z(\cvec(y_{k-1},\xvec)^\top \gammavec(\betavec)) & k = 2, \ldots, K-1\\
     1 - F_Z(\cvec(y_{K-1},\xvec)^\top \gammavec(\betavec)) & k = K.\\
    \end{cases}
    \end{split}
\end{align}
For countably infinite sample spaces (as e.g.\ for count data) with $Y \in \{y_1,y_2,y_3, \ldots\}$, the density is given by
\begin{align}\label{count_ll}
f_Y(y_k |\betavec) &=\begin{cases}
     F_Z(\cvec(y_1,\xvec)^\top \gammavec(\betavec)) & k = 1\\
     F_Z(\cvec(y_k,\xvec)^\top \gammavec(\betavec)) - F_Z(\cvec(y_{k-1},\xvec)^\top \gammavec(\betavec)) & k >1.
    \end{cases}
\end{align}
\paragraph{Censored responses} The Bayesian conditional transformation model incorporates all forms of random censoring. In the presence of censored observations, only the likelihood has to be adapted while the transformation function remains the same. The likelihood contributions for right-, left, and interval-censored continuous or discrete observations respectively are then given by
\begin{align}
     1 - F_Z(\cvec(\underline{y},\xvec)^\top\gammavec(\betavec)) & \text{ for } y \in (\underline{y}, \infty) & \text{ ``right censored"}\nonumber\\
     F_Z(\cvec(\overline{y},\xvec)^\top\gammavec(\betavec))    & \text{ for } y \in (-\infty, \overline{y}) & \text{ ``left censored"}\label{cens_ll}\\
     F_Z(\cvec(\overline{y},\xvec)^\top\gammavec(\betavec)) - F_Z(\cvec(\underline{y},\xvec)^\top\gammavec(\betavec)) &  \text{ for } y \in (\underline{y}, \overline{y}] & \text{ ``interval censored"}.\nonumber
\end{align}%
It is also possible to adapt  densities for truncated observations \citep{hothorn2018}.

\subsection{Formal definition of BCTMs}\label{sec:formaldef}

We end this section with a formal definition of BCTMs. 
\begin{definition}[BCTM]
The quadruple $\big(\varthetavec,F_Z,\cvec, \pi_\vartheta(\cdot)\big)$ of unknown model parameters $\varthetavec$, a choice for the basis $\cvec$, the reference distribution $F_Z$ and joint prior $\pi_\vartheta$ \textit{is called Bayesian conditional transformation model (BCTM)}.
\end{definition}

\section{Posterior Inference}\label{sec:mcmc}

\subsection{Posterior and estimation via MCMC}

Assuming conditional independence the joint posterior is given by
\begin{align}\label{eq:post_beta}
 p(\betavec, \tauvec^2, \omegavec|\yvec) \propto
 \prod_{i=1}^n f_Y(y_i | \betavec)
 \left[\pi(\beta_0)
 \prod_{j=1}^{J} [\pi(\betavec_j |\tau_{j}^2,\omega_j) \pi(\tau_{j}^2),\pi(\omega_j)]
 \right].
\end{align}
\mbox{To obtain samples from  \eqref{eq:post_beta} we use an MCMC sampler consisting of three alternating steps:}
\begin{itemize}[\underline{Step~1.}]
    \item[\underline{Step~1.}] Sample from $p(\betavec|\tauvec^2, \omegavec, \yvec)$ using the NUTS.
    \item[\underline{Step~2.}] For $j=1,\ldots,J$, sample from $p(\tau_{j}^2|\betavec_j,\yvec)$ using a Gibbs sampler in case of an IG prior or  iteratively weighted least squares (IWLS) proposals in case of SD priors.
   \item[\underline{Step~3.}] For $j=1,\ldots,J$, sample $\omega_j$ with a Gibbs step from its discrete full conditional.
\end{itemize}
The resulting MCMC samples can then be used to estimate the conditional distribution $F_{Y|\mX=\xvec}(y)$ as $\hat{F}_{Y|\mX=\xvec,\yvec}(y) = F_Z(\hat{h}(y|\xvec))$
where, for example, $\hat{h}(y|\xvec)$ is the posterior mean estimate
$\hat{h}(y|\xvec)= \cvec(y,\xvec)^\top\frac{1}{S}\sum_{s=1}^S \gammavec^{[s]}$
with posterior samples $\gammavec^{[1]},\ldots,\gammavec^{[S]}$. Similarly, a posterior mean estimate for $F_{Y|\mX=\xvec}(y)$ can be determined as
$\hat{F}_{Y|\mX=\xvec}=\frac{1}{S}\sum_{s=1}^SF_Z(\cvec(y,\xvec)^\top\gammavec^{[s]}).$
The posterior samples also provide us with the basis of deriving the complete posterior distribution of $h(y|\xvec)$, $F_{Y|\mX=\xvec}(y)$, and any transformation thereof.

\noindent\textbf{Updating the basis coefficients  at \underline{Step 1.}}
Basis coefficients are updated jointly by sampling from the log full conditional
\begin{align*}
\log(p(\betavec |\tauvec^2,\omegavec,\yvec)) \propto \sum_{i=1}^n f_Y(y_i|\betavec) - \frac{1}{2} \betavec^\top \mK(\tauvec^2,\omegavec) \betavec,
\end{align*}
where the first term arises from one of the likelihoods described in Sec.~\ref{sec:trafo_ll} and the second term arises from the Gaussian prior). High dimensionality and strong dependencies among  coefficients (stemming partly from the monotonicity constraints) aggravate sampling from the posterior distribution. This is further exacerbated by the mixed linear-nonlinear dependence of the transformation function on $\betatildevec$, rendering e.g.~random-walk Metropolis algorithms slow and inefficient. One possible remedy lies in including gradient information as done by HMC. This, however, comes with the drawback that two additional tuning parameters (step size $\epsilon$ and  number of leapfrog steps $L$) have to be set manually. To avoid this tricky task, we implement NUTS with dual averaging \citep{Nesterov2009} that uses Hamiltonian principles for efficient exploration of the target distribution of $\betavec$ in an adaptive fashion The adaptive nature of NUTS enables a streamlined estimation process, effectively abolishing the need for costly preliminary tuning runs at the expense of some additional computation time per iteration which is owed mainly to the more sophisticated proposals.

The required gradient of the unnormalized log-posterior of the basis coefficients vector $\betavec$ for continuous responses  is given by 
\begin{align*}
s(\betavec)\equiv\frac{\partial\log(p(\betavec|\tauvec^2,\yvec))}{\partial \betavec}  &=
\sum_{i=1}^n\left[\cvec(y_i,\xvec_i)^\top \mSigma \mC
\frac{f_Y^\prime(y_i|\betavec)}{f_Y(y_i|\betavec)} +
\frac{\cvec^\prime(y_i,\xvec_i)^\top \mSigma \mC}{\cvec^\prime(y_i,\xvec_i)^\top \mSigma\betatildevec}\right] -  \mK \betavec,
\end{align*}
where $\mC$ is a diagonal matrix with entries
$C_{dd} = 1$ if $\tilde{\beta}_{d} = \beta_{d}$,
    $C_{dd} =\exp(\beta_{d})$ otherwise,
and similar expressions can straightforwardly be derived for discrete or censored responses.

Potentially  flat parts of a fitted transformation function based on the reparameterization in Sec.~\ref{sec:generic_ct} demand the parameters $\tilde{\beta}=\exp(\beta)$ to be close to zero and thus the corresponding $\beta$ to approach minus infinity. For NUTS, this does not result in overflow errors, but can lead to divergent transitions in the sampling path and NUTS trees with large tree depth as the different curvatures demand very different step sizes. Using a non-centered parametrization \citep{papaspiliopoulos2007general} as a remedy is not feasible in a straightforward manner, because of the nonlinear transformation in the coefficient vectors.  Instead, we found it helpful to increase the goal acceptance rate, forcing the sampler to take smaller steps, which is a small price to pay for the non-occurrence of divergencies. If the problem persists it is possible to drop unidentified (i.e.~reparameterized coefficients that should be close to zero) in each iteration judging by the eigenvalues of the matrix square root of the Hessian of the posterior at \eqref{eq:post_beta} in an efficient way \citep{pyawood}. 

Furthermore, we resort to  augmented precision matrices, e.g.~$\mK_{j} =\frac{1}{\tau_{j}^2}[\omega_j\mK_{1j} + (1-\omega_j)\mK_{2j}] + 10^{-6}\mI$ to ensure positive definiteness and therefore a soft threshold for  coefficient variances \citep{andrinopoulou2018improved}. The NUTS warm-up phase can often be supported by standardizing each covariate or by rescaling them to (0,1). Both measures can facilitate mass matrix adaption. Regarding sampling efficiency, we found that using SD priors for the smoothing variances can decrease run times and improve the effective sample size.

\noindent\textbf{Updating the smoothing variances at \underline{Step 2.}}
When using an IG prior for the smoothing variances, they can be updated directly with a Gibbs step from the full conditional
%\begin{align*}
$\tau_{j}^2 | \cdot \sim \mathrm{IG}\left(a_{j} + \frac{\mathrm{rk}(\mK_{j})}{2}, b_{j} + \frac{1}{2} \betavec_j^\top \mK_{j} \betavec_{j} \right).$
%\end{align*}
For the SD prior, updates can be implemented via IWLS proposals of log-variances following  \citet{KleKne2016}. 

\noindent\textbf{Updating the weights at \underline{Step 3.}}
The updates of the weights are straightforward using Gibbs sampling due to their discrete prior structure \citep{kneib2019modular}.

\noindent\textbf{Computational details}
While BCTMs are pretty robust regarding the choice of  hyperparameters, varying them can improve computational speed and stabilize estimates that involve a monotonicity constraint. All results shown in Secs.~\ref{sec:sims}, \ref{sec:apps} were obtained with $4,000$ MCMC iterations with a NUTS warm-up phase of $2,000$ and a burn-in  of  $2,000$.
 Computations were carried out in R version 4.1.0 \citep{Rlang}. To improve computing time, parts of the sampler were programmed using  \texttt{Rcpp} \citep{Rcpp}. The \texttt{MASS} matrix adaption scheme was adopted from  \texttt{adnuts} \citep{adnuts}. 

\subsection{Estimation of the cCDF}

The resulting $S$ MCMC samples  can  be used to estimate the cCDF $F_{Y|\mX=\xvec}(y)$ as $\hat{F}_{Y|\mX=\xvec,\yvec}(y) = F_Z(\hat{h}(y|\xvec))$
where, for example, $\hat{h}(y|\xvec)$ is the posterior mean estimate
$\hat{h}(y|\xvec)= \cvec(y,\xvec)^\top\frac{1}{S}\sum_{s=1}^S \gammavec^{[s]}$
with posterior samples $\gammavec^{[1]},\ldots,\gammavec^{[S]}$. Similarly, a posterior mean estimate for $F_{Y|\mX=\xvec}(y)$ can be determined as
$\hat{F}_{Y|\mX=\xvec}=\frac{1}{S}\sum_{s=1}^SF_Z(\cvec(y,\xvec)^\top\gammavec^{[s]}).$
The posterior samples also provide us with the basis of deriving the complete posterior distribution (including credible intervals) of $h(y|\xvec)$, $F_{Y|\mX=\xvec}(y)$, and any transformation thereof.

\subsection{Model choice and variable selection }\label{sec:model_sel}
%Cross-validation (CV) can easily become computationally expensive, especially if the data set is large or inference is based on MCMC. Therefore it is useful to consider approximations to cross-validation as means to model selection. 
For model selection, we use the  Watanabe-Akaike information criterion \citep[WAIC][]{watanabe2010}. It can be seen as approximation to computationally expensive cross validation (CV) and is conveniently computed from $s=1, \ldots, S$ posterior samples. We validated WAIC against CV in some of our applications and found good agreements that support using information criteria as the basis for model choice and variable selection.

The WAIC overcomes certain limitations of the DIC \citep[DIC;][]{spiegelhalter2002} such as its dependence on the posterior mean as a specific point estimate or the potential of observing negative effective parameter counts. It is given by
$
\text{WAIC} = (-2l_{\text{WAIC}} + 2p_{\text{WAIC}}),
$
where $l_{\text{WAIC}} = \sum_{i=1}^n \left( \frac{1}{S} \sum_{s=1}^S f_Y(y_i| \betavec^{[s]})\right)
$ and
$ p_{\text{WAIC}} = \sum_{i=1}^n \text{Var}(\log(f_Y(y_i|\betavec))).
$ In the regression literature, information criteria like the DIC and the WAIC are primarily used to discriminate between different types of response distributions and predictor specifications \citep[e.g.][]{klein2015bayesian}. In the holistic approach of BCTMs, the transformation function determines both the response distribution and  the ``predictor'' which is why it is sufficient to use information criteria to compare different (partial) transformation function specifications that differ in flexibility and interaction structure. We also considered the deviance information criterion (DIC) as introduced by \cite{spiegelhalter2002} which yielded similar results and is therefore omitted in the following.

\section{Simulations}\label{sec:sims}

We conducted simulations to evaluate the empirical performance of BCTMs to recover the true data generating process compared to several competing methods from the literature (Sec.~\ref{sec:sim1}) and to provide valid uncertainty estimates by means of coverage rates (Sec.~\ref{sec:sim2}). %evaluate the empirical performance of BCTMs based on the simulation designed in \cite{hothorn2014,hothorn2018} and an additional setup that concentrates on nonlinear shift effects. All simulations shown in the following were repeated in $100$ replications.

\subsection{Recovering the conditional distribution}\label{sec:sim1}

In this section, we mimic the simulation design of \citet{hothorn2014} to benchmark our BCTM against its frequentist counterpart, the MLT as implemented in the R-package  \citep[\textbf{mlt},][]{hothorn2017package}, Bayesian GAMLSS \citep{klein2015bayesian}, and Bayesian semiparametric quantile regression \citep{waldmann2013bayesian}. For both Bayesian benchmarks, we use the R package \textbf{bamlss} \citep{umlauf2018bamlss}.

\noindent\textbf{Simulation design}
For datasets of size $n=200$, we generate two covariates  as i.i.d. realizations via $x_1 \sim U[0,1]$ as well as $x_2 \sim U[-2,2]$. The response $y$ is assumed to follow a heteroscedastic varying coefficient model (VCM)
\begin{align}\label{sim_vcm}
Y &= \frac{1}{x_1 + 0.5} x_2 + \frac{1}{x_1 +0.5}\epsilon, \quad \epsilon \sim \mathrm{N}(0,1),
\end{align}
such that an appropriate CTM has to emulate a Gaussian location-scale model under the premises that the mean depends on the nonlinear varying coefficient $(x_1+0.5)^{-1}$ for $x_2$ and that the variance is a nonlinear function of $x_1$. To analyse the stability in the presence of noise variables, we consider six scenarios, where $p = 0, \ldots, 5$ i.i.d.~realizations from the standard uniform $U[0,1]$ with zero influence on the response are added. The complete vector of covariates is denoted by $\xvec_p = (x_1, x_2, x_3, \ldots, x_{p+2})^\top$.

\noindent\textbf{Benchmark methods} For each of the resulting six scenarios, we fit 
\begin{itemize}
\item Lin.~BCTM $\big( \varthetavec, \Phi,((1,y) \otimes (1,\xvec_p^\top))^\top,\allowbreak \pi_\vartheta(\varthetavec)\big)$: a restricted BCTM consisting of simple linear interactions
\item Lin.~MLT: a linear MLT of the same type 
\item Full BCTM $\big(\varthetavec,\Phi,(\avec(y)^\top \otimes ( \bvec(x_1)^\top, \ldots,\allowbreak \bvec(x_{p+2})^\top))^\top, \pi_\vartheta(\varthetavec)\big)$: a nonlinear BCTM consisting of nonlinear interactions
 with  basis dimension of $10$  in $\avec$ and $\bvec$
\item Full MLT: a nonlinear MLT of the same type with Bernstein polynomials of order $10$, i.e.~with joint basis  $(\avec_{\text{Bs}}(y)^\top \otimes ( \bvec_{\text{Bs}}(x_1)^\top, \ldots, \bvec_{\text{Bs},10}(x_{p+2})^\top))^\top$ 
\item Oracle BAMLSS: a Gaussian location-scale BAMLSS based on model \eqref{sim_vcm}, i.e. $\eta_{\mu} = \beta_0 + x_2 \cdot f(x_1) + \sum_{k=0}^p f(x_{2+k})$ and $\eta_{\sigma^2}= \beta_0 + f(x_1)$  and
\item BAMLSS QR: a Bayesian semiparametric quantile regression specification with nonlinear effects of all explanatory variables
\end{itemize}
Further details on the specifications are given in Supp.~Tab.~C.7.
%The simple parametrization in the linear (B)CTM results from the fact that, though being nonlinear on the scale of the response, the true effects of the covariates are linear on the scale of the transformation function. This information is in general not available which is why the nonlinear (B)CTM does not assume a specific functional form. In that case, both the BCTM and the MLT allow the covariates to influence all moments of the response distribution in a nonlinear fashion.
It is important to stress that  Lin.~BCTM/MLT and Oracle BAMLSS have in common that they are restricted by design to the true (Gaussian) distribution. In addition, the Oracle BAMLSS is the only model that is supplemented with the true predictor for the variance in all scenarios. Yet, despite being linear in the covariates on the scale of the transformation function, the Lin.~BCTM/MLT are nonlinear on the scale of the response. Since this information is in general not available, we also include the Full BCTM/MLT. %Both the nonlinear BCTM and MLT are not restricted to the true distribution family as they are full conditional transformation models based on nonlinear interactions that can resemble a very wide range of continuous distributions.

\noindent\textbf{Performance measures}
As a first measure of performance, we computed the mean absolute deviation (MAD) of the estimates of $F_{Y|\mX=\xvec}(y)$ from the true probabilities over a grid of $y$, $x_1$ and $x_2$
%\begin{align*}
$
\mathrm{MAD}(x_1,x_2) = \frac{1}{n}\sum_{i=1}^n |F_{Y|\mX=\xvec}(y_i) - \hat{F}_{Y|\mX=\xvec}(y_i)|$
%\end{align*}
based on 100 replications. 
Fig.~\ref{fig:sims_all} summarizes the empirical distributions of the minimum, median and maximum MAD for all models that provide estimates for the complete cCDF, i.e.~all but the BAMLSS QR. 
\begin{figure}[htbp]
\centering
   %\makebox[0pt]{\includegraphics[scale=0.9]
      {\includegraphics[width=\textwidth]
   {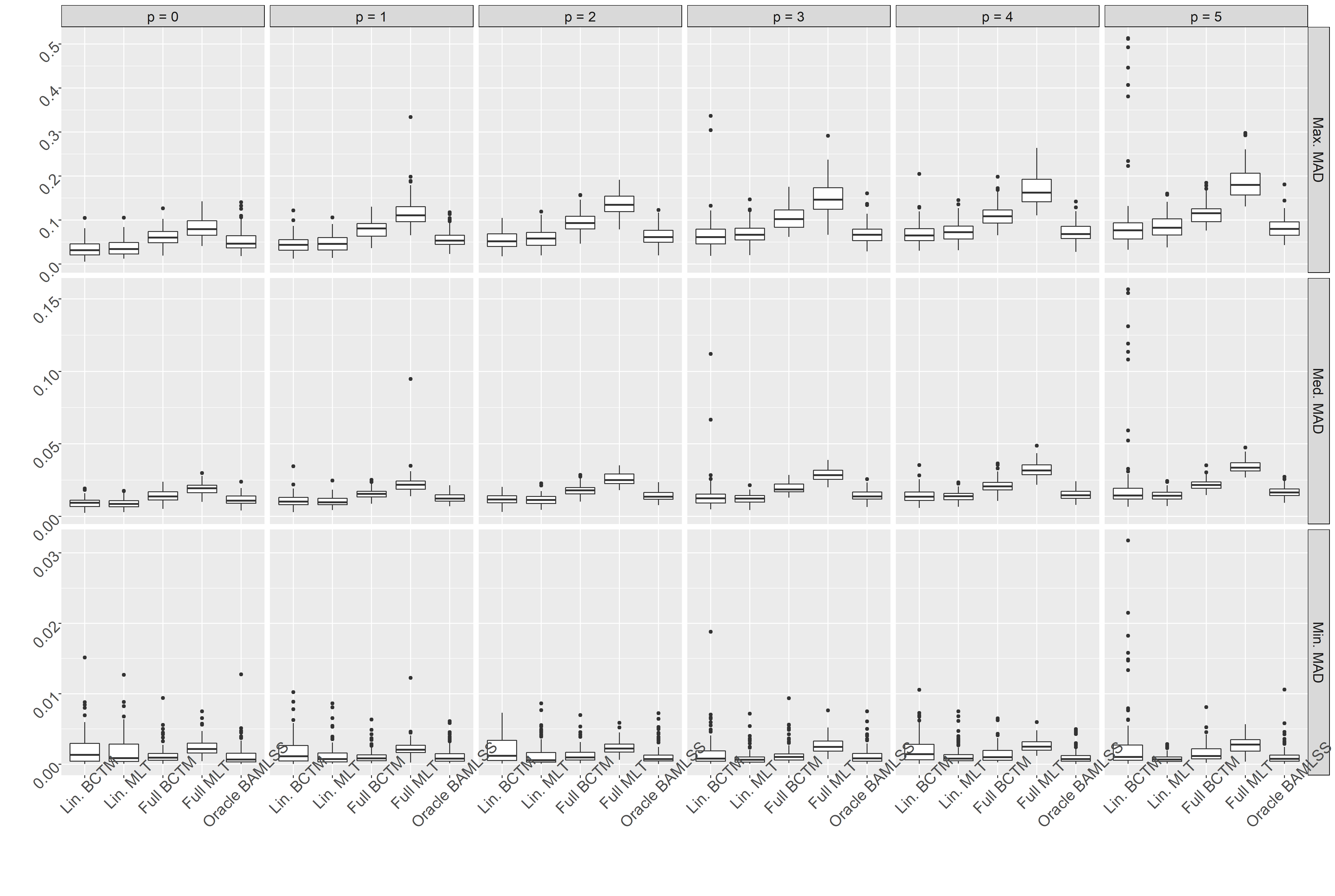}}\caption{Simulation 1. Minimum, median and maximum of the mean absolute deviation (MAD) between true and estimated predictive probabilities for the Lin.~BCTM/MLT, the Full BCTM/MLT and the BAMLSS  for the six scenarios with $p=0,\ldots,5$  noise covariates based on 100 replications each.}\label{fig:sims_all}
\end{figure}
As a second performance measure, we computed conditional quantiles of the fitted response distribution corresponding to a sequence of probabilities $\alpha$ via numerical inversion. Fig.~\ref{fig:quantile_devs} shows the deviations of these from their true counterparts together with similar results obtained via QR BAMLSS which was used as a benchmark.
\begin{figure}[htbp]
\centering
   %\makebox[0pt]{\includegraphics[scale=0.9]
      {\includegraphics[width=\textwidth]
   {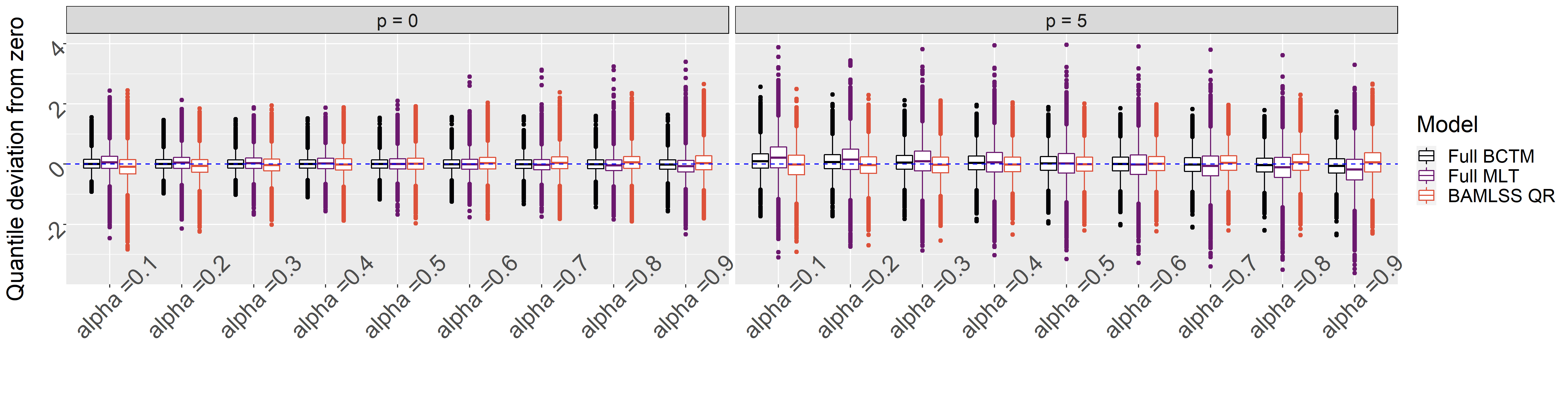}}\caption{\small Simulation 1. Bias in estimated quantiles for various quantile levels $\alpha$ for the six scenarios with $p=0,\ldots,5$  noise covariates based on 100 replications each. }\label{fig:quantile_devs}
\end{figure}
Third, as a measure of accuracy that concentrates on tail features of the distribution, Fig.~\ref{fig:quantile_scores} shows the quantile score function $\text{QS}_\alpha (F^{-1}(\alpha), y) = 2(\mathcal{I}(y < F^{-1} (\alpha)) - \alpha)(\alpha - y)$ at $\alpha=0.05$ and  $\alpha=0.95$,  where $\mathcal{I}(A)$ = 1 if $A$ is true, and zero
otherwise \citep{gneiting2011quantiles}. 
Last, to measure the overall forecast accuracy,  we plot the decomposition of the continuous ranked probability score \citep[CRPS;][]{laio2007verification} which can be written as 
%\begin{align*}
$    \text{CRPS}(F, y) =  \int_0^1 \text{QS}_\alpha(F{-1}(\alpha), y) \mathrm{d}\alpha$
%\end{align*}
in Fig.~\ref{fig:crps}. Both, the QS and CRPS are based on the prediction grids used for the MAD and lower values suggest greater
accuracy.
 \begin{figure}[htbp]
\centering
   %\makebox[0pt]{\includegraphics[scale=0.9]
      {\includegraphics[width=\textwidth]
   {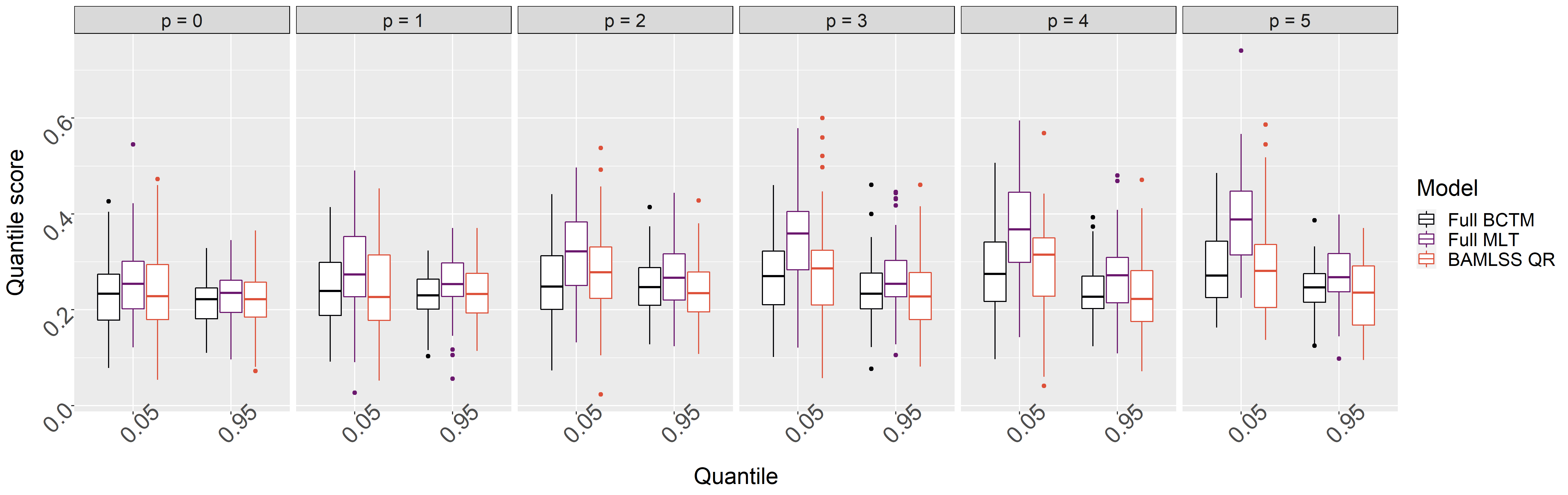}}\caption{\small Simulation 1. Estimated quantile scores QS($F^{-1}(\alpha), y)$ for $\alpha = 0.05$ and $0.95$ for the six scenarios with $p=0,\ldots,5$  noise covariates based on one test dataset each.}\label{fig:quantile_scores}
\end{figure}
\begin{figure}[htbp]
\centering
   %\makebox[0pt]{\includegraphics[scale=0.9]
      {\includegraphics[width=\textwidth]
   {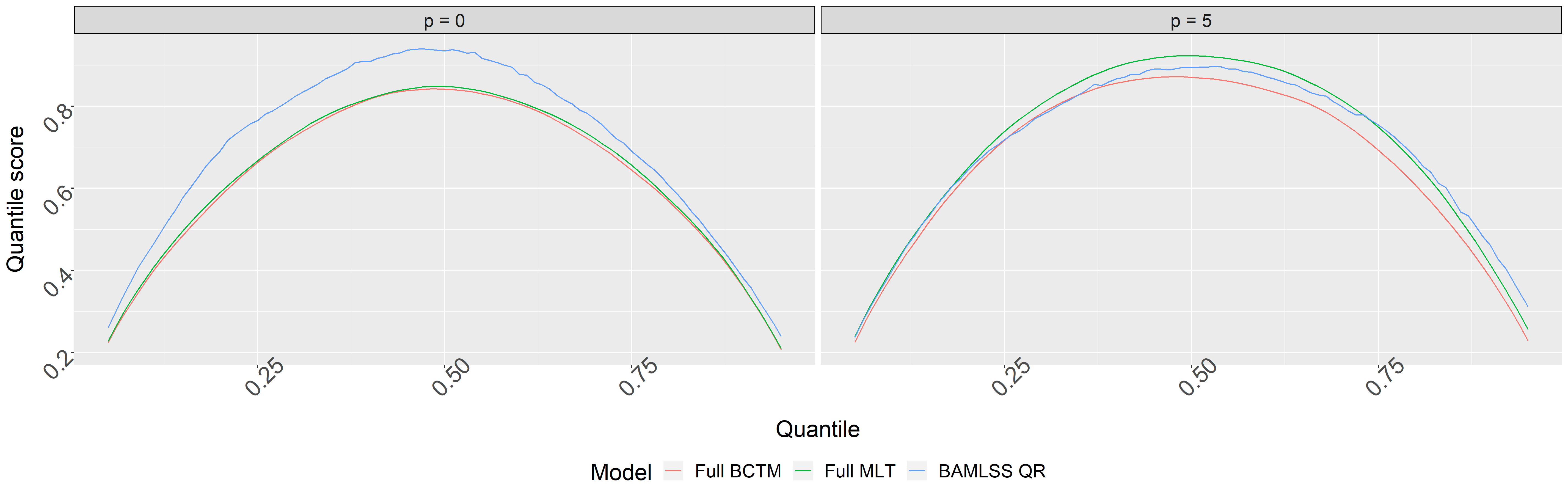}}\caption{\small Simulation 1. Estimation CRPS decomposition  for two scenarios with $p=0,5$  noise covariates based on one test dataset each.}\label{fig:crps}
   \end{figure}

\noindent\textbf{Results}
Lin.~BCTM yields MADs that are very close to those of the Oracle BAMLSS and also performs better than Lin.~MLT for all $p$. The Full BCTM/MLT have somewhat higher MADs, and MLT is again worse than BCTM when the number of noise variables is large. 
All methods do recover the true conditional quantiles well. Full BCTM/MLT are on par with BAMLSS and in particular with the BAMLSS QR which is specifically tailored to estimate conditional quantiles. Full BCTM performs best in terms of QS and CRPS. Specifically, Full BCTM and BAMLSS QR are similar and outperform Full MLT in terms of QS, while full BCTM is slightly better than Full MLT and Bayes QR worst according to the CRPS. In summary, BCTMs enable proper and reliable modelling of the complete cCDF and its quantiles, while avoiding restrictive assumptions on the shape of the  distribution. BCTMs therefore provide a strong competitor in situations where neither the exact predictor specification nor the type of the response distribution can be derived from a priori considerations.

%While in general, all contenders are able to capture the true conditional distribution very well, the linear BCTM and MLT perform better than the BAMLSS in terms of median and maximum MAD, but worse regarding the minimum MAD. All three models perform more and more equally with larger $p$. As expected, the nonlinear BCTM and MLT which are not supplemented with the true distribution family show higher MAD values than their "oracle" counterparts, but by all means hold on nicely. The BCTM shows remarkably strong performance in all investigated cases.

\subsection{Coverage rates}\label{sec:sim2}
To compare BCTM and MLT from a different perspective, we consider empirical coverage rates of pointwise $95\%$ credible/confidence intervals in a simulation setting that concentrates on the estimation of nonlinear covariate effects.

\noindent\textbf{Simulation design} For datasets of size $n=100$ (for $n=500$, see Part C of the Supplement), we  generate four i.i.d.~covariates  via $x_p \sim U[-2,2]$, $p=1,\ldots,4$ and assume four nonlinear test functions $f_1(x) = x$, $f_2(x) = x + \frac{(2x-2)^2}{5.5}$, $f_3(x) = -x + \pi \mathrm{sin}(\pi x)$ and $f_4(x) = 0.5x + 15\phi(2(x-0.2))-\phi(x+0.4)$. The responses are then generated as
$y = f_1(x_1) + f_2(x_2) + f_3(x_3) + f_4(x_4) + \epsilon$, $\epsilon \sim \mathrm{N}(0,1).$

\noindent\textbf{Benchmark methods}  We  fit  linear Gaussian CTMS, i.e.
\begin{itemize}
\item $\big(\varthetavec,\Phi,((1, y)^\top , (\bvec(x_1)^\top, \ldots, \bvec(x_{4})^\top))^\top, \pi_\vartheta(\cdot)\big)$: a linear BCTM with $20$ B-spline basis functions in $\bvec$ and nonlinear shift effects 
\item a linear MLT with nonlinear shifts of the same type specified in terms of Bernstein polynomials of order $10$, i.e.~with joint basis
$((1, y)^\top, ( \bvec_{\text{Bs},10}(x_1)^\top, \ldots,\allowbreak \bvec_{\text{Bs},10}(x_{4})^\top))^\top$
\end{itemize}

\noindent\textbf{Performance measure}
Empirical coverage rates of pointwise 95\% credible/confidence intervals based on 100 replications are shown in Fig.~\ref{fig:cov}. For the BCTM, these can readily be computed from the MCMC output, while for the MLT an additional computationally costly parametric bootstrap \citep{hothorn2017package} has to be run.

\begin{figure}
\centering
   %\makebox[0pt]{\includegraphics[scale=0.9]
      {\includegraphics[width=0.99\textwidth]
   {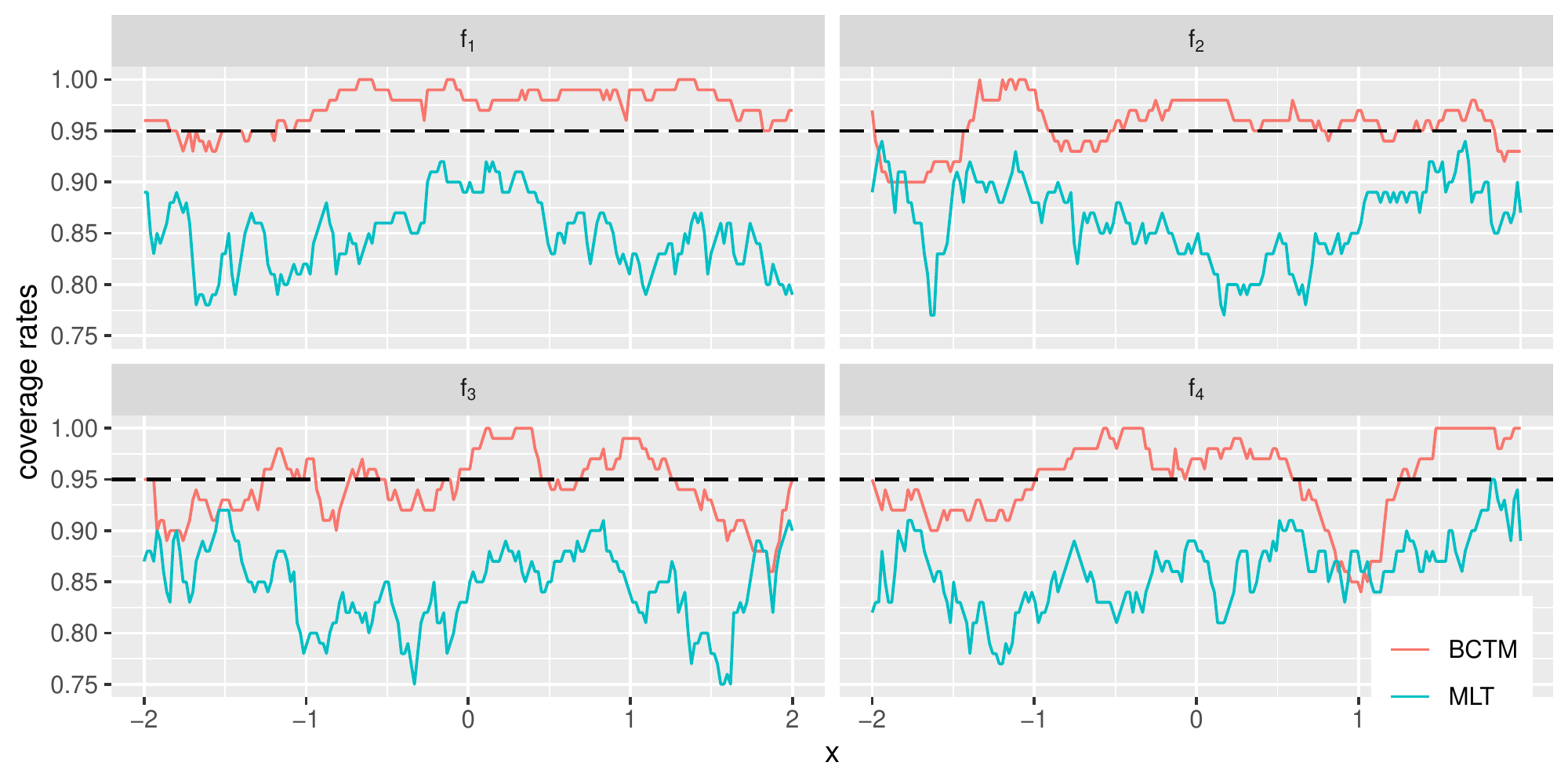}}\caption{Simulation 2. Coverage rates of pointwise 95\% credible/confidence intervals of BCTM (red) and MLT (blue) for $f(x_1), \ldots, f(x_4)$ evaluated on equally spaced grids within the range of $x_1, \ldots, x_4$. The nominal 95\% level is shown as a dashed line.}\label{fig:cov}
\end{figure}

\noindent\textbf{Results}
Fig.~\ref{fig:cov} confirms the validity of the credible intervals provided by the BCTM as the desired 95\% level is mostly maintained which is not the case for  MLT. Corresponding effect estimates are shown in the Supplement, Part C.2.

\section{Applications}\label{sec:apps}

We illustrate the versatility of BCTMs in three applications that differ with respect to the chosen reference distribution and transformation effect types. The first one highlights the applicability of the BCTM in the presence of highly skewed data (Sec.~\ref{sec:fram}). The remaining two are BCTMs for (right-censored) survival data in form of a (non-)proportional hazards (NPH) model with random or spatial frailties
(Sec.~\ref{sec:lung}) and a partial (non-)proportional odds (PO) model (Supplement, Part B.3). While not shown here, it is straightforward to derive additional quantities of interest such as quantile curves or odds by transformations of $\hat{h}$. Throughout this section, we use cubic B-spline bases of dimension $D_1=20$ for $\avec$ for univariate splines and dimension $D_1=10$ for bivariate splines. We adapt the number of basis functions for the covariate effects in $\bvec(\xvec)$ according to subject-matter. As a default, we use IG priors for the smoothing variances unless explicitly stated otherwise.

\subsection{Framingham heart study}\label{sec:fram}
The Framingham Heart Study dataset of  \citep{zhang2001linear} contains the cholesterol levels ($\mathit{cholst}$) of $200$ patients at three to six different measurement points over the course of up to $10$ years along the current $\mathit{age}$ and $\mathit{sex}$ of each individual. There are $n=1044$ observations in total.

We fitted various BCTM specifications that differ in their specific form of the transformation function and the chosen hyperprior. All of them are based on
\begin{itemize}
\item $\big( \varthetavec,\Phi,\allowbreak (\avec(y)^\top \otimes\allowbreak (1, \mathit{age}),\allowbreak (\mathit{year},\allowbreak \mathit{sex}) )^\top, \pi_\vartheta(\varthetavec)\big)$: a response-varying VCM for $\mathit{age}$ and intercept in $\avec(y)$, leading to %, which can be perceived as a Bayesian version of distribution regression \citep{chernozhukov2013inference}. The model is then given by
\begin{align*}
\mathbb{P}(\mathit{cholst} \leq y | \xvec) &= \Phi \left( h_1(y) + h_2(y |\mathit{age})+ h_3(y |\mathit{year})+ h_4(y |\mathit{sex})\right)  \\
&= \Phi\left( \avec(y)^\top \gammavec_1 + \mathit{age}\cdot\bm{a}(y)^\top \gammavec_2  + \mathit{year} \cdot\gamma_3 + \mathit{sex}\cdot \gamma_4  \right).
\end{align*}
\item $\big( \varthetavec,\Phi, ( \avec(y)^\top \otimes  \bvec(\mathrm{age})^\top, (\mathit{year}, \mathit{sex}))^\top,\allowbreak \pi_\vartheta(\varthetavec)\big)$: a full BCTM for $\mathit{age}$, where $\avec(y)$ and $\bvec(\mathit{age})$ contain an intercept, the tensor product is centered around zero and $\bvec(\mathit{age})$ consists of a 10-dimensional B-splines basis leading to 
\begin{align*}
\mathbb{P}(\mathit{cholst} \leq y | \xvec) &= \Phi \left(h_1(y |\mathit{age})+ h_2(y |\mathit{year})+ h_3(y |\mathit{sex})\right)  \\
&= \Phi\left( (\avec(y)^\top \otimes \bvec_1(\mathit{age})^\top)^\top \gammavec_1 + \mathit{year} \cdot\gamma_2 + \mathit{sex}\cdot \gamma_3\right).
\end{align*}
%Note that , the tensor product is centered around zero, and $\bvec$ consists of a cubic 10-dimensional B-spline bases.
\end{itemize}
The default uses IG priors for both models but variants also employ the SD priors as competitors. Furthermore, we considered both models augmented by patient-specific i.i.d.\ random effects. We benchmark the BCTMs against the Bayesian GAMLSS of \citet{MicKleKne2018} based on a skew-t distribution for the responses and predictors $\eta_k$ for all $K=4$ distributional parameters (location, scale, degrees of freedom, and skewness)  given by
\begin{align*}
    \eta_k  = \beta_{k,0}  + x_{\text{sex}} \beta_{k,\text{sex}}  + x_{\text{age}} \beta_{k,\text{age}}  + x_{\text{year}} \beta_{k,\text{year}}, \quad k=1,\ldots,K.
\end{align*}
A variant thereof also contains the patient-specific i.i.d.\ random effects, see Supp.~Tab.~B.1 for full details on all model specifications.

\noindent\textbf{Model selection}
All models are compared to each other using the DIC, WAIC and log-scores in Supp.~Tab.~B.2. The log-scores are based on 10-fold CV. Overall, all criteria  favour the tensor product spline BCTM with random effect over the GAMLSS specifically tailored to skewed responses. In general, the inclusion of random effects seems essential for obtaining realistic models while only smaller improvements result from the consideration of tensor products rather than VCMs. Replacing IG priors with SD priors does only yield  a small performance improvement for the models without random effects. However, applying the SD prior results in noticeabe improvements in effectiveness and stability of the sampler, see Supp.~Tab.~B.3, B.4.

\noindent\textbf{Results}
\begin{figure}
\centering
   \hspace*{-1.5cm}\includegraphics[width=0.7\textwidth]{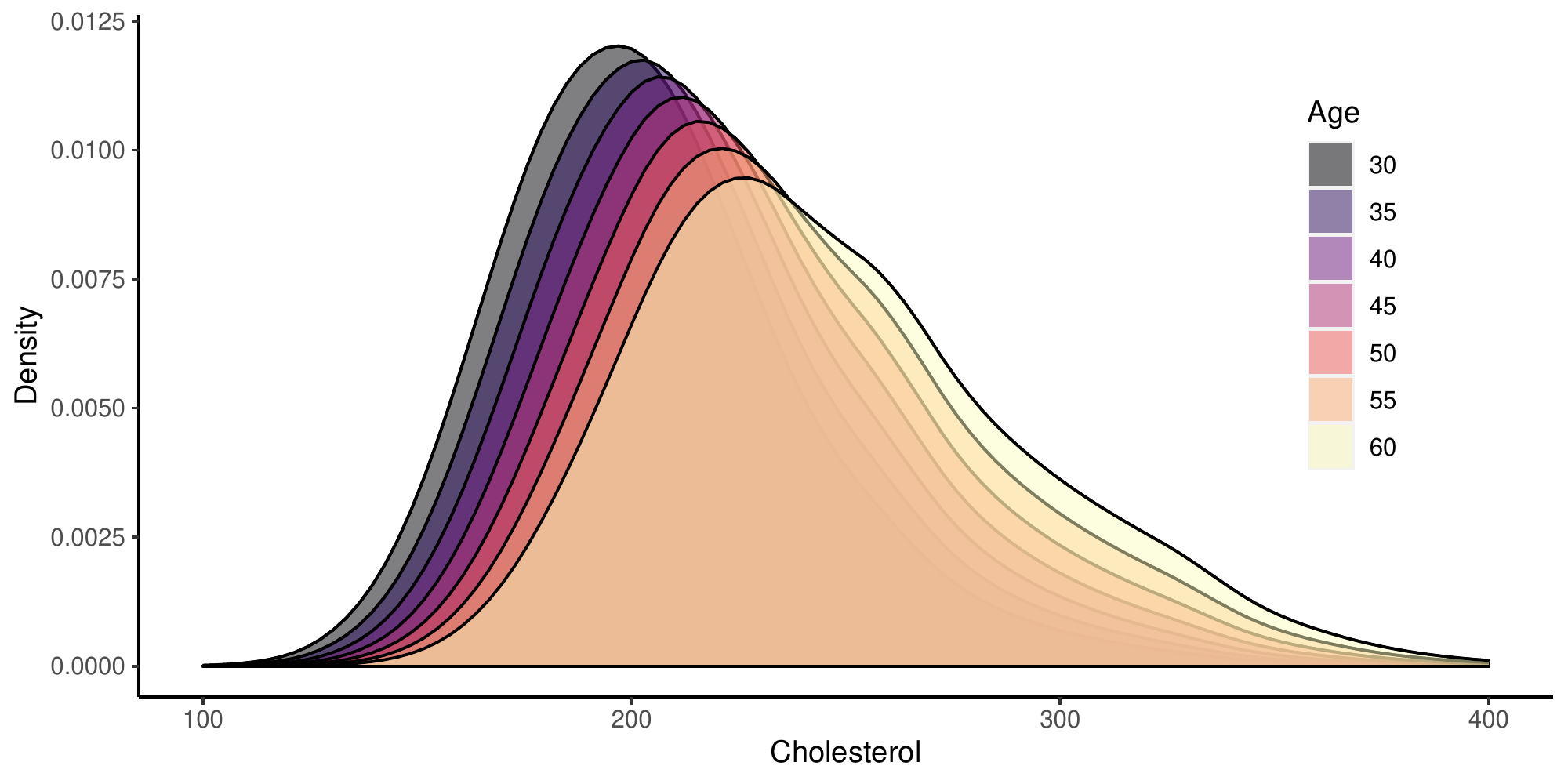}\caption{\small Framingham heart study. VCM with random effects set to zero. Shown are the estimated conditional cholesterol densities for different patient ages in the middle of the study.}\label{fig:fram3}
\end{figure}
Fig.~\ref{fig:fram3} shows estimated conditional densities for different patient ages in the middle of the study ($\mathit{year} = 4$) for the VCM with random effects set to zero. With increasing age, the mode of the conditional distribution is shifted towards higher cholesterol values. Moreover, the estimated conditional densities become more and more right-skewed, indicating the presence of more extreme cholesterol values. On the other hand, the left tail does not change as  much. Fig.~\ref{fig:fram4} shows an estimated heat map that was obtained from the tensor product model assuming nonlinear covariate effects. While the general result is similar to the one in Fig.~\ref{fig:fram3}, we see a reversal of the trend towards right-skewness at $\mathit{age}\approx 55$.
%nk: DAS HABen wir schon in 2.2 gesagt. As further variation, we examined the impact of replacing the IG with the SD hyperprior, where hyperparameter $\theta$ is elicitated on basis of threshold $c=3$ and prior probability $\alpha=0.01$. 
\begin{figure}
\centering
  % \hspace*{-1.5cm}\includegraphics[scale=0.75]{fram_tensor_levelplot.eps}\caption{\small Framingham, tensor todel - Estimated predictive densities of cholesterol for different patient ages in the middle of the study.}\label{fig:fram4}
     \hspace*{-1.5cm}\includegraphics[width=0.7\textwidth]{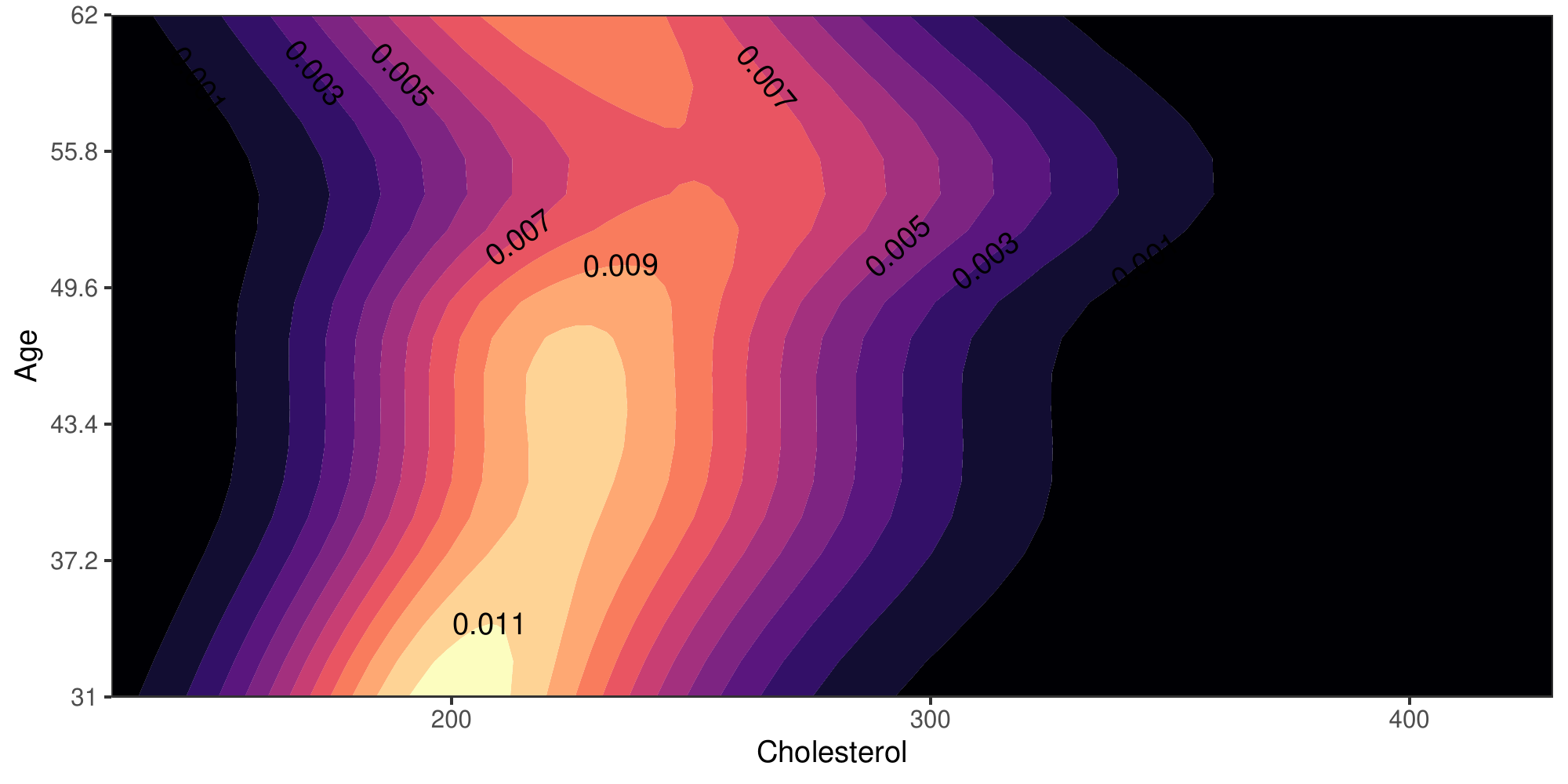}\caption{\small Framingham heart study. Tensor product BCTM with random effects set to zero. Shown are the estimated predictive densities of cholesterol for different patient ages at $\mathit{year}=4$.}\label{fig:fram4}
   \end{figure}

\subsection{Leukemia survival}\label{sec:lung}

The second analysis considers acute myeloid leukemia survival of $n=1043$ patients \citep{henderson2002modeling}  with 184 survival times being right-censored. In addition to the impact of the prognostic factors $\mathit{age}$, $\mathit{sex}$, white blood cell count ($\mathit{wbc}$) and the Townsend score ($\mathit{tpi}$), indicating less affluent residential areas for higher values, we investigate spatial patterns in form of the  indicator $\bvec(s)$ for $24$ administrative regions in North West England. 

In a first step, we fitted  linear PH models $\bctm{F_{\text{MEV}}}{(\avec(t)^\top \otimes 1)^\top, \xvec^\top))^\top}$ both without and with random effect for the administrative districts, and where $F_{\text{MEV}}$ denotes the CDF of the minimum extreme value distribution. 
Next, to account for spatial dependencies through a GMRF, we fit a $\bctm{F_{\text{MEV}}}{(\avec(t)^\top \otimes 1)^\top,(1 \otimes (\bvec(s)^\top, \xvec^\top))^\top}$  resulting in the spatial proportional hazards model
\begin{align*}
P(\mathit{Time} \leq t |\xvec) &= F_{\text{MEV}}(h_1(t)  + h_2(s) + h_3(\xvec))= F_{\text{MEV}}(\avec(t)^\top \gammavec_1  +\xvec^\top \gammavec_2 + \bvec(s)^\top \gammavec_3).
\end{align*}
As a last expansion, we fitted  non-spatial and spatial NPH models for age, i.e. ~$\bctm{F_{\text{MEV}}}{(\avec(t)^\top \otimes \bvec(age))^\top, \xvec^\top))^\top}$ and $\bctm{F_{\text{MEV}}}{(\avec(t)^\top \otimes \bvec(age))^\top,(1 \otimes (\bvec(s)^\top, \xvec^\top))^\top}$, respectively.  All model specifications are compared via the WAIC in Tab.~\ref{tab:leuk:waic} and respective posterior mean estimates of the log-negative harzard ratios are presented in Tab.~\ref{tab:leuk2}. 
\begin{table}
\centering
\begin{tabular}{l|cccccc}
  \hline  \hline
Model & bctm  & bctm\_re & bctm\_spat & bctm\_nph & bctm\_nph\_re  & bctm\_nph\_spat \\  \hline
WAIC &  12425 & 12424 & \textbf{12421} & 12785 & 12768 & 12765\\
  \hline  \hline
\end{tabular}\caption{\small  Leukemia survial. Shown are the WAIC for all BCTM models, i.e.~PH model without random effects (bctm/mlt), with random effects (bctm\_re/mlt\_re), the spatial PH model (bctm\_spat) and the NPH models without and with spatial effect (bctm\_nph/bctm\_nph\_spat).} \label{tab:leuk:waic}
\end{table}
\begin{table}
\centering
\begin{tabular}{l|cccc}
  \hline  \hline
 Model & $\mathit{tpi}$ & $\mathit{age}$ & $\mathit{sex}$ & $\mathit{wbc}$ \\
    \hline
bctm & 0.112 & 0.556 & 0.027 & 0.203 \\
mlt & 0.102 & 0.552 & 0.035 & 0.204 \\
bctm\_re & 0.115 & 0.577 & 0.029 & 0.206  \\
mlt\_re & 0.120 & 0.605 & 0.033 & 0.207   \\
bctm\_spat & 0.114 & 0.590 & 0.035 & 0.208  \\
bctm\_nph & 0.111 & - & 0.028 & 0.198   \\
mlt\_nph & 0.141 & - & 0.005 & 0.190 \\
bctm\_nph\_re & 0.085 & - & 0.043 & 0.202   \\
bctm\_nph\_spat & 0.110 & - & 0.036 & 0.201 \\
   \hline  \hline
\end{tabular}\caption{\small  Leukemia survial. Estimated posterior of the log negative hazard ratios. The rows correspond to the PH model without random effects (bctm/mlt), with random effects (bctm\_re/mlt\_re), the spatial PH model (bctm\_spat) and the NPH models without and with random/spatial effect (bctm\_nph/mlt\_nph/bctm\_nph\_re/bctm\_nph\_spat).} \label{tab:leuk2}
\end{table}
As a baseline check,  Tab.~\ref{tab:leuk2} also includes esimates from the  MLT for which however only the PH model with and without random effects for the districts and the NPH model (without random effects) can be estimated using the R packages \textbf{tram} \citep{tram} and \textbf{tramME} \citep{tamasi2022tramme}. Details on the specifications can be found in Supp.~Part B.2. 
 Estimates of these models are similar to the ones of the corresponding BCTM.

Since overall the WAIC favours the spatial PH model  (bctm\_spat), Fig.~\ref{fig:leuk_plot1} shows the resulting estimated conditional survivor functions defined as
\begin{align*}
S(t|\mX=\xvec) = P(Time>t|\mX=\xvec) = 1-P(Time \leq t |\mX=\xvec)
\end{align*}
for different Townsend scores (Panel A) accompanied by a depiction of the estimated spatial effect (Panel B). It confirms the findings of Tab.~\ref{tab:leuk2}, indicating that affluency (lower $tpi$) is associated with higher survival at all times. The spatial effect of association to a district is associated with lower survival for higher values, and is therefore hinting on a lower mortality cluster in the northwest and on a high-risk ``belt" running from northeast to southwest. 
\begin{figure}
\centering
   \hspace*{-1.5cm}\includegraphics[width=0.8\textwidth]{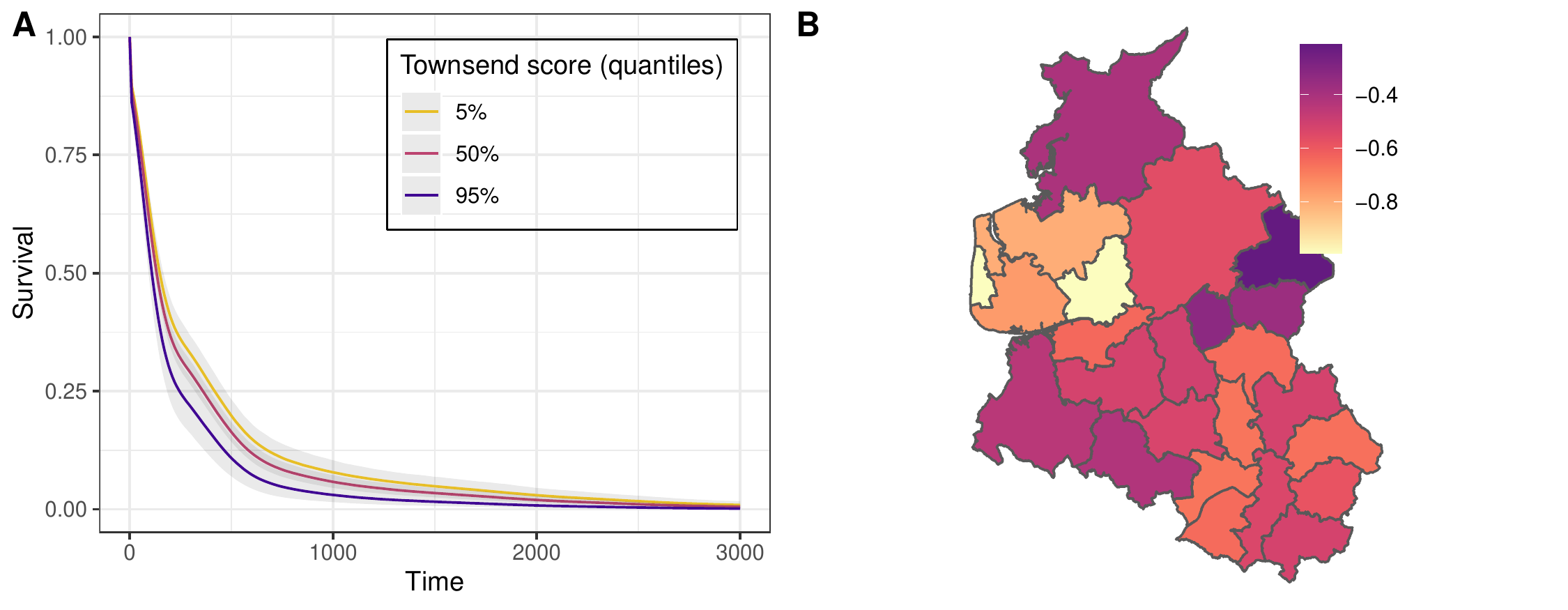}\caption{\small Leukemia survival. (A) Estimated survivor functions with $95\%$ credible interval bands for different quantiles of Townsend score (the lower, the more affluent) for females with $\mathit{age}=49$ and $\mathit{wbc}=38.6$(mean) for the first region. (B) Posterior mean spatial frailties (the larger, the higher mean mortality). Results are based on the spatial PH model $\bctm{F_{\text{MEV}}}{(\avec(t)^\top \otimes 1)^\top, (1 \otimes (\bvec(s)^\top, \xvec^\top))^\top}$.}\label{fig:leuk_plot1}
\end{figure}
Finally, Tab.~B.5 of the Supplement shows the estimated posterior means of the log-negative hazard ratios (collected in $\gammavec_2$), medians and credible intervals of the same model. Similar to the results in \cite{spbayessurv}, we find that $\mathit{tpi}$, $\mathit{age}$ and $\mathit{wbc}$ are significant risk factors for surviving leukemia. 
%the right part of Tab.~\ref{tab:leuk2}. This table also contains (first part) the estimated posterior mean estimates of the log negative hazard ratios. 

% \begin{table}
% \centering
% \begin{tabular}{l|rrrr|rr}
%   \hline  \hline
%  Model & $\mathit{tpi}$ & $\mathit{age}$ & $\mathit{sex}$ & $\mathit{wbc}$ & AIC & WAIC \\
%     \hline
% bctm & 0.112 & 0.556 & 0.027 & 0.203 & 12450 & 12425\\
% mlt & 0.102 & 0.552 & 0.035 & 0.204 & 12708 & -\\
% bctm\_re & 0.115 & 0.577 & 0.029 & 0.206 & - & 12424  \\
% mlt\_re & 0.120 & 0.605 & 0.033 & 0.207 & - & -  \\
% bctm\_spat & 0.114 & 0.590 & 0.035 & 0.208 & - & \textbf{12421} \\
% bctm\_nph & 0.111 & - & 0.028 & 0.198 & 12964 & 12785  \\
% mlt\_nph & 0.141 & - & 0.005 & 0.190 & 13198 & -\\
% bctm\_nph\_re & 0.085 & - & 0.043 & 0.202 & - &  12768 \\
% bctm\_nph\_spat & 0.110 & - & 0.036 & 0.201 & - & 12765 \\
%   \hline  \hline
% \end{tabular}\caption{\small  Leukemia survial. Estimated posterior of the log negative hazard ratios and information criteria. The rows correspond to the PH model without random effects (bctm/mlt), with random effects (bctm\_re/mlt\_re), the spatial PH model (bctm\_spat) and the NPH models without and with random/spatial effect (bctm\_nph/mlt\_nph/bctm\_nph\_re/bctm\_nph\_spat).} \label{tab:leuk2}
% \end{table}

\section{Summary and discussion}\label{sec:summary}
Our Bayesian treatment of CTMs based on MCMC is attractive for an assortment of reasons. Sampling-based inference provides posterior samples for coefficients of the conditional transformation function which can be transferred to samples of the cCDF, but also to samples of any quantity of interest that relies on the cCDF. It is straightforward, for example, to obtain point estimates and credible intervals without having to dive into asymptotics. Furthermore, the Bayesian paradigm offers a natural way to impose smoothness penalties on the crucial nonlinear transformation functions, a feature likelihood-based competing methods (such as the MLT) are lacking in the software. In this way, the BCTM is able to resemble and even expand upon models ranging from simple to complex in settings with continuous, discrete and censored data without requiring strong assumptions.

In flat regions of the curve however, the reparameterization of the basis coefficients is such that the untransformed $\betavec$ may be weakly identified, resulting in potentially inefficient sample runs. This issue is explicitlty tackled in \cite{mckay2011variable} who use a spike and slab prior directly on the basis coefficients that resulted in zero coefficients for flat regions, but only allows nonlinear monotonic covariate effects in a Gaussian setting. In our approach, we considered scale-dependent hyperpriors to counter mixing problems, but expanding such investigations to a wider scope with interactions is certainly an interesting field for future research.

%The modularity of BCTMs grants high flexibility in the model building process which facilitates  expansions in many directions . One drawback of using IG priors for the smoothing variances is that they do in general not place sufficient probability mass on values close to zero, which makes them prone to overfitting in some cases. Supplementing the BCTM with SD priors as described by \cite{KleKne2016} together with a prior elicitation scheme \citep{klein2020bayesian} could lead to even better function estimates. This is further exacerbated by the observation that our penalty smooths towards a straight line  since in this case, the base model consists of the Gaussian family characterized by linear transformations which could be a promising leverage point towards an even deeper understanding of the model. Here, the smoothing variance acts as a distance measure controlling proximity to the Gaussian family.

Instead of specifying $F_Z$ a priori, it could also be interesting to include it as an additional free parameter in the estimation process. Among others, \cite{linton2008,politis2013} describe the situation of a ``model free'' paradigm where the reference distribution is estimated without invoking any predetermined model (but by fully parameterizing the transformation function). This restriction is alleviated by the fact that in theory, arbitrarily complex distributions can be transformed to a basic reference distribution as long as the transformation function $h$ is flexible enough. Abandoning the additivity assumption in $h(y|\xvec)$ in favor of e.g.~tensor spline interactions however, can become computationally costly and numerically unstable due to the high dimensionality of the resulting basis, but can also be tackled by estimating the reference distribution in conjunction with a simpler structure in the transformation function.  The idea of a free $F_Z$ was investigated in a Bayesian setting by \cite{walker1999bayesian} and \cite{mallick2003bayesian} for example who use a P\'{o}lya tree prior for a series of (unconditional) semiparametric transformation models. Embedding it in the BCTM framework would result in a potentially very powerful addition to model flexibility.

Finally, our occupation with the BCTM for this article and beyond assured us that it represents an alluring modern competitor in the race to capture more and more aspects of the response distribution beyond the mean.

\bigskip
\begin{center}
{\large\bf SUPPLEMENTARY MATERIAL}
\end{center}
\spacingset{1.0}
\begin{description}

\item[supplement.pdf] This supplement contains the proof of Theorem~\ref{theo1} and further results for simulations and applications.
\item[Code]  to reproduce the results from the applications is available on request. % NK: ersetze das spaeter durch "upon request" oder github?
\end{description}

\spacingset{1.1}
%\small
\bibliography{lit}
\end{document}